\documentclass[11pt]{article}

\usepackage{algorithm}
\usepackage{algorithmic}
\usepackage{latexsym}
\usepackage{graphicx}
\usepackage{amsmath}

\newtheorem{theorem}{Theorem}
\newtheorem{lemma}[theorem]{Lemma}
\newtheorem{definition}[theorem]{Definition}
\newenvironment{proof}{\textit{Proof:}~}{\hfill$\Box$\par\vskip1em}

\newcommand{\TOWERN}[1]{{ \hspace{-1.2mm} \begin{array}{c} {#1} \end{array} \hspace{-1.2mm}}}
\newcommand{\TOWERA}[2]{{ \hspace{-1.2mm} \begin{array}{c} {#1}\\ {#2} \end{array} \hspace{-1.2mm}}}
\newcommand{\TOWERB}[3]{{ \hspace{-1.2mm} \begin{array}{c} {#1}\\ {#2} \\ {#3} \end{array} \hspace{-1.2mm}}}
\newcommand{\TOWERC}[4]{{ \hspace{-1.2mm} \begin{array}{c} {#1}\\ {#2} \\ {#3} \\ {#4} \end{array} \hspace{-1.2mm}}}

\newcommand{\EMP}{\emptyset}
\newcommand{\RIGHT}{\rightarrow}
\newcommand{\LEFT}{\leftarrow}
\newcommand{\XRIGHT}[1]{\xrightarrow{#1}}
\newcommand{\BOTH}{\leftarrow \vee \rightarrow}
\newcommand{\CONF}{{\cal C}}
\newcommand{\RULE}{{\cal R}}
\newcommand{\ALG}{{\cal A}}
\newcommand{\PROB}{{\cal P}}
\newcommand{\TER}{{\cal T}}
\newcommand{\W}{{\textsf W}}
\newcommand{\G}{{\textsf G}}

\title{Ring Exploration with Myopic Luminous Robots}

\author{
Fukuhito Ooshita\\
Nara Institute of Science and Technology,\\
Graduate School of Science and Technology,\\
Takayama 8916-5, Ikoma, Nara, Japan\\
\texttt{f-oosita@is.naist.jp}\\
\and
S\'{e}bastien Tixeuil\\
Sorbonne Universit\'{e}, CNRS,\\
Laboratoire d'Informatique de Paris 6, LIP6,\\
FR-75005, Paris, France\\
\texttt{Sebastien.Tixeuil@lip6.fr}
}

\begin{document}

\maketitle

\begin{abstract}
We investigate exploration algorithms for autonomous mobile robots evolving in uniform ring-shaped networks. Different from the usual Look-Compute-Move (LCM) model, we consider two characteristics: myopia and luminosity. Myopia means each robot has a limited visibility. We consider the weakest assumption for myopia: each robot can only observe its neighboring nodes. Luminosity means each robot maintains a non-volatile visible light. We consider the weakest assumption for luminosity: each robot can use only two colors for its light. The main interest of this paper is to clarify the impact of luminosity on exploration with myopic robots.

As a main contribution, we prove that 1) in the fully synchronous model, two and three robots are necessary and sufficient to achieve perpetual and terminating exploration, respectively, and 2) in the semi-synchronous and asynchronous models, three and four robots are necessary and sufficient to achieve perpetual and terminating exploration, respectively. These results clarify the power of lights for myopic robots since, without lights, five robots are necessary and sufficient to achieve terminating exploration in the fully synchronous model, and no terminating exploration algorithm exists in the semi-synchronous and asynchronous models. 

We also show that, in the fully synchronous model (resp., the semi-synchronous and asynchronous models), the proposed perpetual exploration algorithm is universal, that is, the algorithm solves perpetual exploration from any solvable initial configuration with two (resp., three) robots and two colors. On the other hand, we show that, in the fully synchronous model (resp., the semi-synchronous and asynchronous models), no universal algorithm exists for terminating exploration, that is, no algorithm may solve terminating exploration from any solvable initial configuration with three (resp., four) robots and two colors.

\noindent {\bf Keywords:} autonomous mobile robots, deterministic exploration, discrete environments, limited visibility, visible light
\end{abstract}

\section{Introduction}

\subsection{Background and Motivation}
Studies about cooperation of autonomous mobile robots have attracted a lot of attention recently in the field of Distributed Computing. The main goal of those works is to characterize the minimum capabilities of robots that permit to achieve a given task. Since the pioneering work of Suzuki and Yamashita~\cite{Suzuki99:Distributed}, many results have been published in their Look-Compute-Move (LCM) model. In the LCM model, each robot repeats executing cycles of look, compute, and move phases. At the beginning of each cycle, the robot observes positions of other robots (look phase). According to its observation, the robot computes whether it moves somewhere or stays idle (compute phase). If the robot decides to move, it moves to the target position by the end of the cycle (move phase). To consider minimum capabilities, most studies assume that robots are identical (\emph{i.e.}, robots execute the same algorithm and cannot be distinguished), oblivious (\emph{i.e.}, robots have no memory of their past actions), and silent (\emph{i.e.}, robots cannot communicate with other robots explicitly). Indeed, communication among robots is done only in an implicit way by observing positions of other robots and moving to a new position. Previous works considered problem solvability of LCM robots in continuous environments (aka two- or three-dimensional Euclidean space) ~\cite{Flocchini05:Gathering,Fujinaga15:Pattern,Suzuki99:Distributed,Yamauchi17:Plane}, while others considered discrete environments (aka graph networks) ~\cite{DAngelo17:Unified,DAngelo14:Gathering,Flocchini13:Computing,Klasing10:Taking,Klasing08:Gathering}.

In this paper, we focus on robots evolving in graph networks. The most fundamental tasks in graph networks are gathering and exploration. The goal of gathering is to make all robots meet at a non-predetermined single node. Gathering has been studied for rings~\cite{DAngelo17:Unified,DAngelo14:Gathering,Klasing10:Taking,Klasing08:Gathering}, grids and trees~\cite{DAngelo16:Gathering}.
Two types of exploration tasks have been well studied: perpetual exploration requires robots to visit nodes so that every node is visited infinitely many times by a robot, and terminating exploration requires robots to terminate after every node is visited by a robot at least once. For example, perpetual exploration has been studied for rings~\cite{Blin10:Exclusive} and grids~\cite{Bonnet11:Asynchronous}, and terminating exploration has been studied for rings~\cite{Devismes13:Optimal,Flocchini13:Computing}, trees~\cite{Flocchini10:Remembering}, grids~\cite{Devismes12:Optimal}, tori~\cite{Devismes15:Optimal}, and arbitrary networks~\cite{Chalopin10:Network}. 

All aforementioned works in graph networks make the assumption that each robot observes all other robots in the networks. That is, each robot has a sensor that can obtain a global snapshot. However, this powerful ability somewhat contradicts the principle of very weak mobile entities. For this reason, recent studies consider the more realistic case of myopic robots~\cite{Datta13:Ring,Datta15:Enabling,Guilbault13:Gathering,Guilbault13:Gatheringline}. A myopic robot has limited visibility, \emph{i.e.}, it can see nodes (and robots on them) only within a certain fixed distance $\phi$. Datta et al.\ studied terminating exploration of rings for $\phi=1$~\cite{Datta13:Ring} and $\phi=2,3$~\cite{Datta13:Ring23}. Guilbault and Pelc studied gathering in bipartite graphs with  $\phi=1$~\cite{Guilbault13:Gathering}, and infinite lines with $\phi>1$~\cite{Guilbault13:Gatheringline}. Not surprisingly, in the weakest setting, \emph{i.e.}, $\phi=1$, robots can only achieve few tasks. It is shown~\cite{Datta13:Ring} that, when $\phi=1$ holds, five robots are necessary and sufficient to achieve terminating exploration in the fully synchronous (FSYNC) model. On the other hand, no terminating exploration algorithm exists in the semi-synchronous (SSYNC) and asynchronous (ASYNC) models. Also,  gathering~\cite{Guilbault13:Gathering} is possible when $\phi=1$ only if robots initially form a star.

Since most results for myopic robots with $\phi=1$ are negative, a natural question is which additional assumptions can improve task solvability. In this paper, we focus on a non-volatile visible light~\cite{Das16:Autonomous} as an additional assumption. A robot endowed with such a light is called a luminous robot. Each luminous robot is equipped with a light device that can emit a constant number of colors to other robots, a single color at a time. The light color is non-volatile, so it can be used as a constant-space memory. For non-myopic luminous robots, the power of lights is well understood~\cite{Das16:Autonomous,DEmidio16:Characterizing,Heriban18:Optimally}. For example, if each robot has a five colors light, the difference between the asynchronous model and the semi-synchronous model disappears~\cite{Das16:Autonomous}. However, to the best of our knowledge, the impact of lights on myopic robots has not been studied yet.

\subsection{Our Contributions}
We focus on ring exploration and  the impact of lights on myopic robots with $\phi=1$. We consider the weakest assumption for lights: each robot can use only two colors for its light. Table \ref{table:summary} summarizes our contributions and related works. Note that robots with no light are equivalent to robots with a single color light.

\begin{table}[t]
\centering
\caption{Ring exploration with myopic robots.}
\label{table:summary}
\begin{tabular}{|c|c|c|c|c|c|c|}
\hline
 & & & & & \multicolumn{2}{|c|}{\#robots} \\ \cline{6-7}
Reference & Exploration & Synchrony & $\phi$ & \#colors & necessary & sufficient \\ \hline
\cite{Datta13:Ring} & terminating & FSYNC & 1 & 1 & 5 & 5 \\ \hline
\cite{Datta13:Ring} & terminating & SSYNC \& ASYNC & 1 & 1 & \multicolumn{2}{|c|}{impossible} \\ \hline
\cite{Datta15:Enabling} & terminating & SSYNC \& ASYNC & 2 & 1 & 5 & 7 \\ \hline
\cite{Datta15:Enabling} & terminating & SSYNC \& ASYNC & 3 & 1 & 5 & 5 \\ \hline
This paper & perpetual   & FSYNC & 1 & 2 & 2 & 2 \\ \hline
This paper & terminating & FSYNC & 1 & 2 & 3 & 3 \\ \hline
This paper & perpetual   & SSYNC \& ASYNC & 1 & 2 & 3 & 3 \\ \hline
This paper & terminating & SSYNC \& ASYNC & 1 & 2 & 4 & 4 \\ \hline
\end{tabular}
\end{table}

As a main contribution, we prove that \emph{(i)} in the fully synchronous model, two and three robots are necessary and sufficient to achieve perpetual and terminating exploration, respectively, and \emph{(ii)} in the semi-synchronous and asynchronous models, three and four robots are necessary and sufficient to achieve perpetual and terminating exploration, respectively. These results clarify the power of lights for myopic robots since, without lights, five robots are necessary and sufficient to achieve terminating exploration in the fully synchronous model, and no terminating exploration algorithm exists in the semi-synchronous and asynchronous models. Interestingly, even if robots can observe nodes up to distance three (\emph{i.e.}, $\phi=3$), five robots are required to achieve terminating exploration without light. This means that there exist some tasks that myopic luminous robots with small visibility can achieve, but that non-luminous robots with larger visibility cannot.

Similarly to previous works for myopic robots, all algorithms proposed in this paper assume some specific initial configurations because most configurations are not solvable. For example, when myopic robots are deployed so that no robot can observe other robots, they cannot achieve exploration. However, our perpetual exploration algorithms achieve the best possible property, that is, they are universal. This means that, in the fully synchronous model (resp., the semi-synchronous and asynchronous models), the proposed algorithm solves perpetual exploration from any solvable initial configuration with two (resp., three) robots and two colors. As for terminating exploration, we show that no universal algorithm exists. That is, in the fully synchronous model (resp., the semi-synchronous and asynchronous models), no algorithm may solve terminating exploration from any solvable initial configuration with three (resp., four) robots and two colors. 

Due to space limitation, we omit some of proofs. The omitted proofs are given in the appendix.

\section{Preliminaries}

\subsection{System model}
The system consists of $n$ nodes and $k$ mobile robots. 
The nodes $v_0, v_1, \ldots, v_{n-1}$ form an undirected and unoriented ring-shaped graph, where a link exists between $v_i$ and $v_{i+1}$, for $i<n$, and between $v_{n-1}$ and $v_0$. 
For simplicity we consider mathematical operations on node indices as operations modulo $n$. Neither nodes nor links have identifiers or labels, and consequently robots cannot distinguish nodes and links. 
Robots do not know $n$, the size of the ring. 
Robots occupy some nodes of the ring. The distance between two nodes is the number of links in a shortest path between the nodes. The distance between two robots $a$ and $b$ is the distance between two nodes occupied by $a$ and $b$. Two robots $a$ and $b$ are neighbors if the distance between $a$ and $b$ is one. A set $S$ of robots is \emph{connected} if the induced subgraph of nodes occupied by the robots in $S$ is connected; otherwise, $S$ is disconnected.

Robots we consider have the following characteristics and capabilities. Robots are {\it identical}, that is, robots execute the same deterministic algorithm and cannot be distinguished based on their appearance (in particular, their do \emph{not} have unique identifiers). Robots are {\it luminous}, that is, each robot has a light (or state) that is visible to itself and other robots. A robot can choose the color of its lights from a discrete set $Col$. When the set $Col$ is finite, we denote by $\kappa$ the number of available colors (\emph{i.e.}, $\kappa=|Col|$). 
Robots have no other persistent memory and cannot remember the history of past actions. Robots cannot communicate with other robots explicitly, however they can communicate implicitly by observing positions and colors of other robots (for collecting information), and by changing their color and moving (for sending information). Each robot $r$ can observe positions and colors of robots within a fixed distance $\phi$ ($\phi>0$) from its current position. Since robots are identical, they share the same $\phi$. If $\phi=\infty$, robots can observe all other robots in the ring. If $\phi=1$, robots are \emph{myopic}, that is, they can only observe robots that are located at neighboring nodes.

Each robot executes an algorithm by repeating three-phases cycles: Look, Compute, and Move (L-C-M). During the \emph{Look} phase, the robot observes positions and colors of robots within distance $\phi$. During the \emph{Compute} phase, the robot computes its next color and movement according to the observation in the Look phase. The robot may change its color at the end of the Compute phase. If the robot decides to move, it moves to a neighboring node during the \emph{Move} phase. 
To model asynchrony of executions, we introduce the notion of {\it scheduler} that decides when each robot executes phases. When the scheduler makes robot $r$ execute some phase, we say the scheduler activates the phase of $r$ or simply activates $r$. We consider three types of synchronicity: the FSYNC (full-synchronous) model, the SSYNC (semi-synchronous) model, and the ASYNC (asynchronous) model. In the FSYNC model, the scheduler executes full cycles of all robots synchronously and concurrently. In the SSYNC model, the scheduler selects a non-empty subset of robots and executes full cycles of the selected robots synchronously and concurrently. In the ASYNC model, the scheduler executes cycles of robots asynchronously. Note that in the ASYNC model, a robot $r$ can move based on an outdated view observed previously by $r$. Throughout the paper we assume that the scheduler is {\it fair}, that is, each robot is activated infinitely often. We consider the scheduler as an adversary. That is, we assume that the scheduler is omniscient (it knows robot positions, colors, algorithms, etc.), and tries to activate robots in such a way that they fail executing the task. 

In the sequel, $M_i(t)$ denotes the multiset of colors of robots located in node $v_i$ at instant $t$. If $v_i$ is not occupied by any robot at $t$, then $M_i(t)=\emptyset$ holds, and $v_i$ is \emph{free} at instant $t$. Then, $v_i$ is a \emph{tower} at instant $t$ if $|M_i(t)|\ge 2$.
A {\it configuration} $C(t)$ of the system at instant $t$ is defined as $C(t)=(M_0(t),M_1(t),\ldots,M_{n-1}(t))$. If $t$ is clear from the context, we simply write $C=(M_0,M_1,\ldots,M_{n-1})$.
If there exists an index $x$ such that $M_{x+i}=M_{x-i}$ holds for any $i$, or if $M_{x+i}=M_{x-(i+1)}$ holds for any $i$ (\emph{i.e.}, there exists at least one axis of symmetry in the configuration), configuration $C$ is called {\it symmetric}. 

When a robot observes its environment, it gets a {\it view} up to distance $\phi$. Consider a robot $r$ on node $v_i$; then, $r$ obtains two views: the forward view and the backward view. The forward and backward views of $r$ are defined as $V_f=(c_{r},M_{i-\phi},\ldots,M_{i-1},M_i,M_{i+1},\ldots,M_{i+\phi})$, and $V_b=(c_{r},M_{i+\phi},\ldots,M_{i+1},M_i,M_{i-1},\ldots,M_{i-\phi})$, respectively, where $c_{r}$ denotes $r$'s color.
Since we assume unoriented rings (where robots may not share the same notion of left and right), each robot cannot distinguish its forward view from its backward view.
If the forward view and the backward view of $r$ are identical, then $r$'s view is \emph{symmetric}. In this case, $r$ cannot distinguish between the two directions when it moves, and the scheduler decides which direction $r$ moves to. If $r$ observes no other robot in its view, $r$ is \emph{isolated}.

\subsection{Algorithm, execution, and problem}

An algorithm is described as a set of rules. Each rule is represented in the following manner $<Label>:<Guard>::<Action>$. The guard $<Guard>$ is a possible view obtained by a robot. If a forward or backward view of robot $r$ matches a guard in an algorithm, we say $r$ is enabled. We also say the corresponding rule $<Label>$ is enabled. If a robot is enabled, the robot may change its color and move based on the corresponding action $<Action>$ during the Compute and Move phases. 

For an infinite sequence of configurations $E=C_0,C_1,\ldots,C_t,\ldots$, we say $E$ is an execution from initial configuration $C_0$ if, for every instant $t$, $C_{t+1}$ is obtained from $C_t$ after some robots execute phases. We say $C_i$ is the $i$-th configuration of execution $E$.

A problem $\PROB$ is defined as a set of executions: An execution $E$ solves problem $\PROB$ if $E\in \PROB$ holds. An algorithm $\ALG$ solves problem $\PROB$ from initial configuration $C_0$ if any execution from initial configuration $C_0$ solves problem $\PROB$. We simply say an algorithm $\ALG$ solves problem $\PROB$ if there exists an initial configuration $C_0$ such that $\ALG$ solves $\PROB$ from $C_0$. For configuration $C$ and problem $\PROB$, $C$ is solvable for $\PROB$ if there exists an algorithm (specific to $C$) that solves $\PROB$ from initial configuration $C$. Let $C_s(\PROB)$ be a set of all configurations solvable for $\PROB$. We say algorithm $\ALG$ is universal with respect to problem $\PROB$ if $\ALG$ solves $\PROB$ from any initial configuration in $C_s(\PROB)$.

\subsection{Exploration problems}

In this paper, we consider perpetual exploration problem and terminating exploration problem in case of $\phi=1$.

\begin{definition}[Perpetual exploration problem]

{\it Perpetual exploration} is defined as a set of executions $E$
such that every node is infinitely many times visited by some robot in $E$.
\end{definition}

\begin{definition}[Terminating exploration problem]
{\it Terminating exploration} is defined as a set of executions $E$ such that 1) every node is visited by at least one robot in $E$ and 2) there exists a suffix of $E$ such that no robots are enabled.
\end{definition}

\subsection{Descriptions}
Let $C=(M_0,\ldots,M_{n-1})$ be a configuration. We say $C'=(M'_0,\ldots,M'_{n'-1})$ is a sub-configuration of $C$ if there exists $x$ such that $M_{x+i}=M'_i$ holds for any $i$ ($0\le i\le n'-1$). In this case, we say $n'$ is the length of sub-configuration $C'$. We sometimes describe a sub-configuration $C'=(M'_0,\ldots,M'_{n'-1})$ by listing all colors in $M'_i$ as the $i$-th column. That is, when $M'_i=\{c_1^i,\ldots,c_{|M'_i|}^i\}$ holds for each $i$ ($0\le i\le n'-1$), we describe $C'$ as follows:
\[
\TOWERC{c_1^0}{c_2^0}{\vdots}{c_{|M'_0|}^0}
\TOWERC{c_1^1}{c_2^1}{\vdots}{c_{|M'_1|}^1}
\cdots
\TOWERC{c_1^{n'-1}}{c_2^{n'-1}}{\vdots}{c_{|M'_{n'-1}|}^{n'-1}}
\]
When $M'_i=\emptyset$ holds, we write $\EMP$ as the $i$-th column. If $h$ free nodes exist successively, we sometimes write $\EMP^h$ instead of writing $h$ columns with $\EMP$. For simplicity, when $C'$ is a sub-configuration of $C$ and all robots appear in $C'$, we use $C'$ instead of $C$ to represent configuration $C$. We also use this description to represent views of robots.

Throughout the paper, we consider the case of $\phi=1$. We describe a rule in an algorithm in the following manner:
\[
\RULE_{rule}:\ \ 
\TOWERC{c_{-1,1}}{c_{-1,2}}{\vdots}{c_{-1,m_{-1}}}
\TOWERC{c_{0,1}}{c_{0,2}}{\vdots}{(c_{0,m_{0}})}
\TOWERC{c_{1,1}}{c_{1,2}}{\vdots}{c_{1,m_{1}}}
::\ c_{new},Movement
\]
Notation $\RULE_{rule}$ is a label of the rule. The middle part represents a guard. This represents a view $V=(c_{0,m_{0}},M_{-1},M_0,M_1)$, where $M_i=\{c_{i,1},\ldots,c_{i,m_i}\}$ holds for $i\in\{-1,0,1\}$. Intuitively, each column represents colors of robots on a single node and a color within parentheses represents its current color. If a forward or backward view of robot $r$ is equal to $V$, $r$ is enabled. In this case, $r$ can execute an action represented by $c_{new},Movement$. Notation $c_{new}$ represents a new color of the robot, and $Movement$ represents the movement. Notation $Movement$ can be $\bot$, $\LEFT$, $\RIGHT$, or $\BOTH$: 1) $\bot$ implies a robot does not move, 2) $\LEFT$ (resp., $\RIGHT$) implies a robot moves toward the node such that a set of robot colors is $M_{-1}$ (resp., $M_{1}$), and 3) $\BOTH$ implies a robot moves toward one of two directions (the scheduler decides the direction). When the view $V$ described in a guard is symmetric, $Movement$ should be either $\bot$ or $\BOTH$. As an example, consider the following rule.
\[
\RULE_{ex}:\ \ 
\TOWERA{}{\EMP} \TOWERA{\G}{(\W)} \TOWERA{}{\G} ::\ \G,\RIGHT
\]
Robot $r$ is enabled by $\RULE_{ex}$ if 1) the color of $r$ is $\W$, 2) the current node is occupied by two robots with colors $\G$ and $\W$, 3) one neighboring node is occupied by no robot, and 4) another neighboring node is occupied by a robot with color $\G$. If $r$ is enabled by $\RULE_{ex}$, $r$ changes its color to $\G$ and moves toward the node occupied by a robot with color $\G$.

\section{Full-synchronous Robots}

\subsection{A universal perpetual exploration algorithm for two robots with two colors}

In this subsection, we provide a perpetual exploration algorithm for two robots with two colors. Note that, since one robot cannot achieve perpetual exploration clearly because the direction of its movement is decided by the scheduler, two robots are necessary to achieve perpetual exploration. A set of colors is $Col=\{\G,\W\}$. The algorithm is given in Algorithm \ref{alg:fp2}. In the initial configuration, two robots with colors $\G$ and $\W$ stay at neighboring nodes. In this algorithm, the robot with color $\G$ moves against the other robot, and the robot with color $\W$ moves toward the other robot. This implies two robots move in the same direction. Since they move synchronously, the views of the two robots are not changed. Hence, the two robots continue to move and achieve perpetual exploration. Clearly we have the following theorem.

\begin{theorem}
\label{thm:fp-alg-cor}
In case of $\phi=1$ and $k=2$, Algorithm \ref{alg:fp2} solves perpetual exploration from initial configurations $\G\W$ and $\W\G$ for $n\ge 2$ in the FSYNC model.
\end{theorem}

\begin{algorithm}[t]
\caption{Fully-Synchronous Perpetual Exploration for $k=2$}
\label{alg:fp2}
\begin{algorithmic}
\renewcommand{\algorithmicrequire}{\textbf{Initial configurations}}
\REQUIRE
\STATE $\G\W$ and $\W\G$
\renewcommand{\algorithmicrequire}{\textbf{Rules}}
\REQUIRE
\STATE $0\G\W:\ \ \TOWERN{\EMP}\TOWERN{(\G)}\TOWERN{\W} ::\ \G,\LEFT$
\STATE $0\W\G:\ \ \TOWERN{\EMP}\TOWERN{(\W)}\TOWERN{\G} ::\ \W,\RIGHT$
\end{algorithmic}
\end{algorithm}

In the following, we show that other initial configurations are unsolvable for $n\ge 6$. This implies Algorithm \ref{alg:fp2} is universal with respect to perpetual exploration for $n\ge 6$ in case of $\phi=1$ and $k=2$. To prove the impossibilities, we first prove Lemma \ref{lem:imp-territory}. This lemma will be used for many impossibility proofs not only in the FSYNC model but also in the SSYNC and ASYNC models. For configuration $C$, we define $V_r(C)$ as a set of nodes occupied by at least one robot. We say a set of two neighboring nodes $T=\{v_i,v_{i+1}\}$ is a territory of robots on node $v$ if $v\in T$ holds. We say a territory set $\TER$ is independent if, for every pair of territories $T_1, T_2\in \TER$, the distance between any node in $T_1$ and any node in $T_2$ is at least two (See Fig.\,\ref{fig:territoryset}).

\begin{figure}[t]
    \centering
    \includegraphics[scale=0.35]{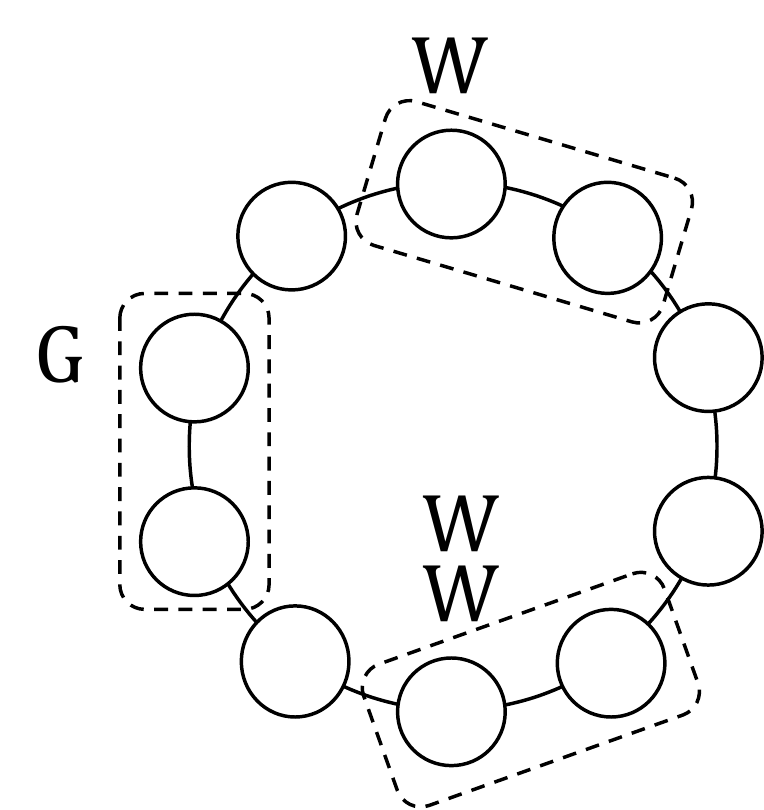}
    \caption{An example of an independent territory set. Notations $\W$ and $\G$ represent robots with colors $\W$ and $\G$, respectively.}
    \label{fig:territoryset}
\end{figure}

\begin{lemma}
\label{lem:imp-territory}
Consider a configuration $C$ such that no two nodes in $V_r(C)$ are neighbors and, for every node $v\in V_r(C)$, robots on $v$ have the same color. If there exists a territory set $\TER$ such that $\TER$ is independent and every node in $V_r(C)$ belongs to some territory in $\TER$, robots on $v\in V_r(C)$ cannot go out of their territory in $\TER$ after configuration $C$ in the FSYNC, SSYNC, and ASYNC models.
\end{lemma}

\begin{proof}
Assume that such a territory set $\TER$ exists. We prove the lemma by induction. At configuration $C$, every robot stays at its territory. Consider configuration $C'$ such that every robot stays at its territory. Since no two nodes in different territories share the same neighbor, each robot observes no robots on its neighbor nodes. This means the view of the robot is symmetric and consequently the scheduler can decide the direction of its movement. In addition, all robots on a single node have the same color, they make the same behaviors if the scheduler activates them at the same time. Hence, each robot moves to another node in its territory if it decides to move. That is, every robot stays at its territory again. Therefore, the lemma holds.
\end{proof}

Now we go back to the FSYNC model, and prove that initial configurations other than $\G\W$ and $\W\G$ are unsolvable.

\begin{lemma}
\label{lem:fp-imp-isolated}
Assume $n\ge 6$. Let $C$ be a configuration such that two robots are disconnected. In this case, $C$ is unsolvable.
\end{lemma}

\begin{proof}
Let $d$ be the distance between two robots in $C$. Without loss of generality, we assume that two robots $r_1$ and $r_2$ occupy nodes $v_0$ and $v_d$ at $C$, respectively. We define territories $T_1$ and $T_2$ as follows: If $d=2$ holds, we define $T_1=\{v_0,v_{n-1}\}$ and $T_2=\{v_d,v_{d+1}\}$, and if $d>2$ holds, we define $T_1=\{v_0,v_1\}$ and $T_2=\{v_d,v_{d+1}\}$. Since $\TER=\{T_1,T_2\}$ is independent, two robots visit only nodes in $T_1\cup T_2$ from Lemma \ref{lem:imp-territory}. Therefore, they cannot achieve perpetual exploration.
\end{proof}

\begin{lemma}
\label{lem:fp-imp-gg}
Assume $n\ge 6$. Configurations $\G\G$ and $\W\W$ are unsolvable.
\end{lemma}

\begin{proof}
Since two robots have the same view, they move in a symmetric manner. If each robot moves toward the other robot, the robots just swap their positions. Hence, to achieve exploration, eventually each robot moves against the other robot. After the movement, the distance between them is three. Similarly to Lemma \ref{lem:fp-imp-isolated}, robots cannot achieve perpetual exploration from the configuration.
\end{proof}

From Theorem \ref{thm:fp-alg-cor} and Lemmas \ref{lem:fp-imp-isolated} and \ref{lem:fp-imp-gg}, a set of solvable configurations is $C_s=\{\G\W,\W\G\}$. Therefore, we have the following theorem.

\begin{theorem}
In case of $\phi=1$, $k=2$, and $Col=\{\G,\W\}$, Algorithm \ref{alg:fp2} is universal with respect to perpetual exploration for $n\ge 6$ in the FSYNC model.
\end{theorem}

\subsection{Impossibility of terminating exploration with two robots}

In this subsection, we prove that no algorithm solves terminating exploration for $k=2$.

\begin{theorem}
\label{thm:no-two-fsync}
In case of $\phi=1$ and $k=2$, no algorithm solves terminating exploration in the FSYNC model. This holds even if robots can use an infinite number of colors.
\end{theorem}

\begin{proof}
Assume that such algorithm $\ALG$ exists. 
Consider an execution $E=C_0,C_1,\ldots$ of $\ALG$ in a $n_1$-node ring $R_1$ ($n_1\ge 6$). Let $t$ be the minimum instant such that two robots terminate or become disconnected at $C_t$. Next, for some $n_2> 2(t+1)$, let us consider an execution $E'=C'_0,C'_1,\ldots$ of $\ALG$ in a $n_2$-node ring $R_2$. Clearly, as long as two robots keep connected, they do not recognize the difference between $R_1$ and $R_2$. Hence, in $E'$, two robots move similarly to $E$ until $C'_t$. If two robots terminate at $C'_t$, they have visited at most $2(t+1)$ nodes and thus they do not achieve exploration. If two robots become disconnected at $C'_t$, we can define an independent territory set at $C'_t$. From Lemma \ref{lem:imp-territory}, two robots cannot visit the remaining nodes and thus they cannot achieve exploration. This is a contradiction.
\end{proof}

\subsection{A terminating exploration algorithm for three robots with two colors}

In this subsection, we give a terminating exploration algorithm for three robots with two colors in case of $n\ge 3$. A set of colors is $Col=\{\G,\W\}$. The algorithm is given in Algorithm \ref{alg:ft3}.

\begin{algorithm}[t]
\caption{Fully-Synchronous Terminating Exploration for $k=3$}
\label{alg:ft3}
\begin{algorithmic}
\renewcommand{\algorithmicrequire}{\textbf{Initial configurations}}
\REQUIRE
\STATE $\W\W\W$, $\G\W\W$, $\W\W\G$, and $\G\W\G$
\renewcommand{\algorithmicrequire}{\textbf{Rules}}
\REQUIRE
\STATE $0\G\W:\ \ \TOWERN{\EMP}\TOWERN{(\G)}\TOWERN{\W} ::\ \G,\LEFT$
\STATE $0\W\G:\ \ \TOWERN{\EMP}\TOWERN{(\W)}\TOWERN{\G} ::\ \W,\RIGHT$
\STATE $0\W\W:\ \ \TOWERN{\EMP}\TOWERN{(\W)}\TOWERN{\W} ::\ \G,\bot$
\STATE $\G\W\W:\ \ \TOWERN{\G}\TOWERN{(\W)}\TOWERN{\W} ::\ \W,\LEFT$
\STATE $\G\W\G:\ \ \TOWERN{\G}\TOWERN{(\W)}\TOWERN{\G} ::\ \W,\BOTH$
\end{algorithmic}
\end{algorithm}

Executions of Algorithm \ref{alg:ft3} for $n\ge 5$ are given in Fig.\,\ref{fig:alg-ft}. We consider three robots $r_1$, $r_2$, and $r_3$. In the figure, $\W_i$ (resp., $\G_i$) represents robot $r_i$ with color $\W$ (resp., $\G$). Arrows represent that indicated robots are enabled. At configuration $\W\W\W$, $r_1$ and $r_3$ are enabled by rule $0\W\W$ (Fig.\,\ref{fig:alg-ft}(a)) and change their colors to $\G$. At configuration $\G\W\G$ (Fig.\,\ref{fig:alg-ft}(b)), robots $r_1$ and $r_2$ (\emph{i.e.}, a pair of robots $\G\W$) and $r_3$ (\emph{i.e.}, another robot $\G$) move to the opposite directions by rules $0\G\W$ and $\G\W\G$. Note that, since the view of $r_2$ is symmetric at configuration $\G\W\G$, the scheduler decides the direction of $r_2$. This implies that the next configuration is $\G\W\EMP\EMP \G$ (Fig.\,\ref{fig:alg-ft}(c)) or $\G\EMP\EMP \W\G$. However, since the two configurations are symmetric to each other, robots move in the same manner after the configuration. At configuration $\G\W\W$ (Fig.\,\ref{fig:alg-ft}(d)), robots $r_1$ and $r_2$ move to the opposite direction of $r_3$ by rules $0\G\W$ and $\G\W\W$, and then the configuration becomes one in Fig.\,\ref{fig:alg-ft}(e). Note that, since configuration $\W\W\G$ is symmetric to $\G\W\W$, robots move in the same manner from configuration $\W\W\G$. After configurations in Fig.\,\ref{fig:alg-ft}(c)(e), robots $r_1$ and $r_2$ continue to move to the same direction by rules $0\G\W$ and $0\W\G$. After $r_1$ and $r_2$ explore the ring, they reach $r_3$ (Fig.\,\ref{fig:alg-ft}(f)). After robot $r_2$ moves, they terminate at a configuration in Fig.\,\ref{fig:alg-ft}(g).

\begin{figure}[t]
\begin{center}
\includegraphics[scale=0.35]{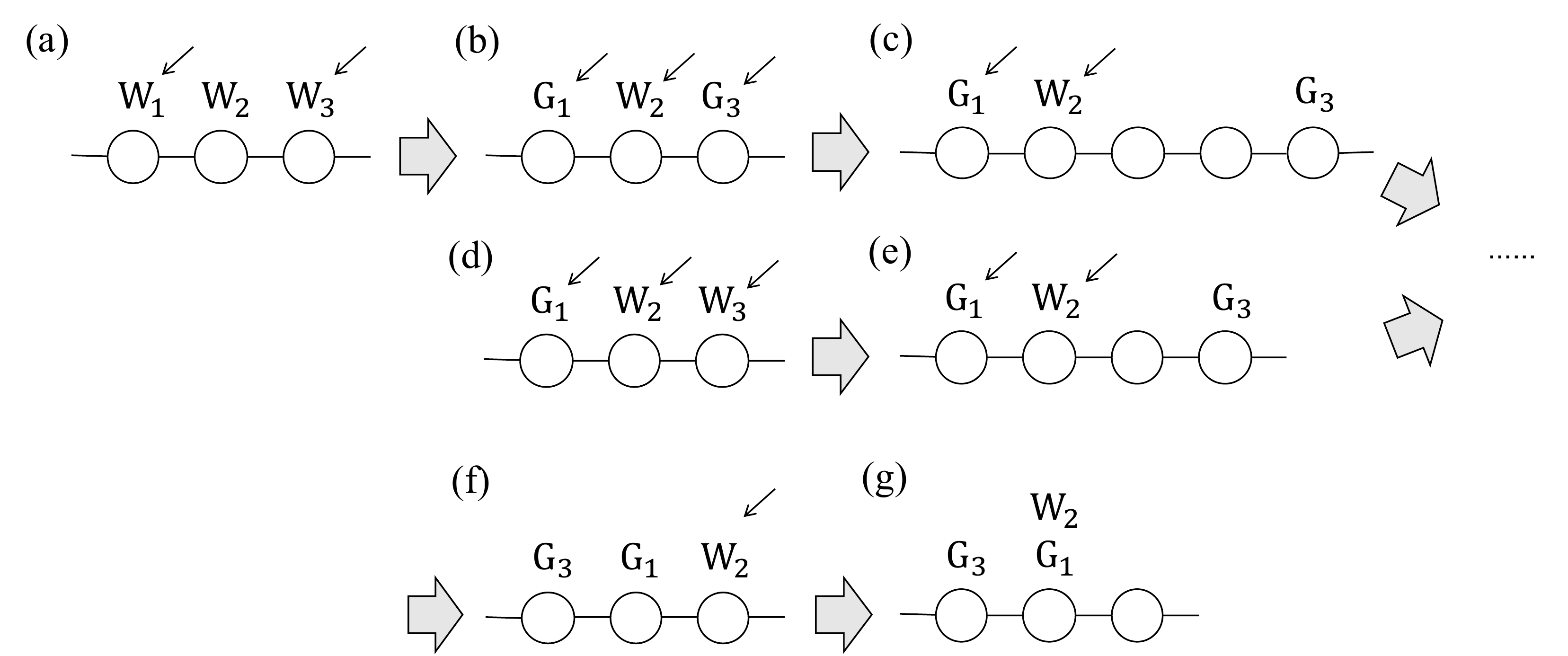}
\caption{Executions of Algorithm \ref{alg:ft3}}
\label{fig:alg-ft}
\end{center}
\end{figure}

We can easily verify that Algorithm \ref{alg:ft3} works for $n=3$ or $n=4$. Hence we have the following theorem.

\begin{theorem}
In case of $\phi=1$ and $k=3$, Algorithm \ref{alg:ft3} solves terminating exploration from initial configurations $\W\W\W$, $\G\W\W$, $\W\W\G$, and $\G\W\G$ for $n\ge 3$ in the FSYNC model.
\end{theorem}

Note that we can construct another algorithm by swapping the roles of colors $\G$ and $\W$ in Algorithm \ref{alg:ft3}. Clearly this algorithm solves terminating exploration from configurations such that colors $\G$ and $\W$ are swapped from solvable configurations for Algorithm \ref{alg:ft3}. This implies configurations $\G\G\G$, $\W\G\G$, $\G\G\W$, and $\W\G\W$ are also solvable. Hence, we have the following lemma.

\begin{lemma}
\label{lem:ft3-solvable}
If $k=3$ holds and a set of colors is $\{\G,\W\}$, configurations $\W\W\W$, $\W\W\G$, $\W\G\W$, $\W\G\G$, $\G\W\W$, $\G\W\G$, $\G\G\W$, and $\G\G\G$ are solvable for terminating exploration in the FSYNC model.
\end{lemma}

We also prove that there exists no universal algorithm with respect to terminating exploration for three robots with two colors. This validates the assumption that Algorithm \ref{alg:ft3} starts from some designated initial configuration.

\begin{theorem}
\label{thm:no-universal-fsync}
In case of $\phi=1$, $k=3$, and $\kappa=2$, no universal algorithm exists with respect to terminating exploration in the FSYNC model.
\end{theorem}

\begin{proof}
(Sketch) To prove the lemma by contradiction, we assume universal algorithm $\ALG$ exists. We first prove that, by $\ALG$, the distance between some pair of two robots becomes at least five in some large ring. Intuitively this holds because, if the three robots always stay within distance four, they cannot distinguish small rings and large rings, and terminate without exploration in large rings.

When the distance becomes at least five, some two robots should be connected. Otherwise, we can define an independent territory set and exploration is impossible from Lemma \ref{lem:imp-territory}. Intuitively the two connected robots should explore the ring because another robot cannot go out of its territory. This implies the two robots should form sub-configuration $\G\W$. In addition, since one of the two robots should move toward another robot to explore the ring, $\ALG$ should include rule $0\W\G X: \EMP (\W) \G:: X,\RIGHT$ or $0\G\W X: \EMP (\G) \W:: X,\RIGHT$ for some $X\in Col$.

Assume $\ALG$ includes rule $0\W\G X$ (we can prove another case similarly). Consider a solvable configuration $\W\G\W$. From this configuration, two robots with color $\W$ move toward $\G$ at the same time by rule $0\W\G X$. This implies two robots with the same color make a tower in the next configuration. After the configuration, since the two robots move in the same manner, three robots behave as if there were only two robots. Similarly to Theorem \ref{thm:no-two-fsync}, we can prove that they cannot achieve terminating exploration. This is a contradiction because $\ALG$ cannot achieve terminating exploration from a solvable initial configuration.
\end{proof}

\section{Semi-synchronous and Asynchronous Robots}
\label{sec:async-perp}

\subsection{Impossibility of perpetual exploration with two robots}

In this subsection, we prove that two robots are not sufficient to achieve perpetual exploration in the SSYNC model. Clearly this impossibility result holds in the ASYNC model.

\begin{theorem}
\label{thm:sp-imp-two}
In case of $\phi=1$ and $k=2$, no algorithm solves perpetual exploration in the SSYNC model. This holds even if robots can use an infinite number of colors.
\end{theorem}

\begin{proof}
Assume that such algorithm $\ALG$ exists for $n\ge 6$. First consider the case that initially two robots $r_1$ and $r_2$ are connected. Clearly some robot $r_i$ ($i\in\{1,2\}$) eventually moves against another robot. At that time, we assume that the scheduler activates only $r_i$. Since another robot does not move, the two robots become disconnected. 

From the above fact, two robots are initially disconnected or eventually become disconnected. Let us consider the configuration such that two robots are disconnected. In this case, we can define an independent territory set, and thus robots cannot achieve exploration from Lemma \ref{lem:imp-territory}. This is a contradiction.
\end{proof}

\subsection{A universal perpetual exploration algorithm for three robots with two colors}

In this subsection, we give a universal perpetual exploration algorithm for three robots with two colors in the SSYNC and ASYNC models. We first give a perpetual exploration algorithm by three robots with two colors in the ASYNC model, and after that we prove the algorithm is universal in the SSYNC and ASYNC models. A set of colors is $Col=\{\G,\W\}$. The algorithm is given in Algorithm \ref{alg:ap3}.

\begin{algorithm}[t]
\caption{Asynchronous Perpetual Exploration for $k=3$}
\label{alg:ap3}
\begin{algorithmic}
\renewcommand{\algorithmicrequire}{\textbf{Initial configurations}}
\REQUIRE
\STATE $\W\W\G$, $\W\G\G$, $\G\W\W$, $\G\G\W$,
 $\TOWERA{}{\W}\TOWERA{\G}{\W}$, $\TOWERA{\G}{\W}\TOWERA{}{\W}$, $\TOWERA{}{\G}\TOWERA{\W}{\G}$, and $\TOWERA{\W}{\G}\TOWERA{}{\G}$.
\renewcommand{\algorithmicrequire}{\textbf{Rules}}
\REQUIRE
\STATE $0\G\W:\ \ \TOWERN{\EMP}\TOWERN{(\G)}\TOWERN{\W} ::\ \G,\RIGHT$ \vspace{3mm}
\STATE $0T\W:\ \ \TOWERA{}{\EMP}\TOWERA{\G}{(\W)}\TOWERA{}{\W} ::\ \G,\RIGHT$ \vspace{3mm}
\STATE $0T\G:\ \ \TOWERA{}{\EMP}\TOWERA{\W}{(\G)}\TOWERA{}{\G} ::\ \W,\LEFT$ \vspace{3mm}
\STATE $0\W\G:\ \ \TOWERN{\EMP}\TOWERN{(\W)}\TOWERN{\G} ::\ \W,\RIGHT$
\end{algorithmic}
\end{algorithm}

Executions of Algorithm \ref{alg:ap3} for $n\ge 4$ are given in Fig.\,\ref{fig:alg-ap}. Let us consider configuration $\W\W\G$, and assume that $r_1$, $r_2$, and $r_3$ compose the configuration in this order (Fig.\,\ref{fig:alg-ap}(a)). Here only $r_3$ is enabled with rule $0\G\W$, and $r_3$ moves toward $r_2$. In a configuration in Fig.\,\ref{fig:alg-ap}(b), only $r_2$ is enabled with rule $0T\W$. If $r_2$ is activated, $r_2$ changes its color to $\G$ and moves toward $r_1$ (Fig.\,\ref{fig:alg-ap}(c)). Note that, in the ASYNC model, after $r_2$ changes its color, some robots may observe the intermediate configuration before $r_2$ moves toward $r_1$. However, since no rule matches the intermediate configuration, robots do not move based on the configuration. After $r_2$ moves from Fig.\,\ref{fig:alg-ap}(c) by rule $0T\G$, the sub-configuration becomes $\W\W\G$ (Fig.\,\ref{fig:alg-ap}(e)) but the robots change their positions from Fig.\,\ref{fig:alg-ap}(a) to Fig.\,\ref{fig:alg-ap}(e). Similarly, robots repeat the behavior from Fig.\,\ref{fig:alg-ap}(a) to Fig.\,\ref{fig:alg-ap}(e), and they achieve perpetual exploration. From configuration $\W\G\G$ in Fig.\,\ref{fig:alg-ap}(d), $r_1$ moves by rule $0\W\G$ and becomes a configuration in Fig.\,\ref{fig:alg-ap}(c). After that, they move similarly to the case from a configuration in Fig.\,\ref{fig:alg-ap}(a). These executions include configurations
\[
\W\W\G, \W\G\G, \TOWERA{}{\W}\TOWERA{\G}{\W}, \ and\  \TOWERA{\W}{\G}\TOWERA{}{\G},
\]
and consequently from these configurations robots can achieve perpetual exploration. Since remaining configurations
\[
\G\W\W, \G\G\W, \TOWERA{\G}{\W}\TOWERA{}{\W}, \ and\  \TOWERA{}{\G}\TOWERA{\W}{\G}
\]
are symmetric to the above configurations, robots can also achieve perpetual exploration from the configurations. Therefore, we have the following theorem.

\begin{figure}[t]
\begin{center}
\includegraphics[scale=0.35]{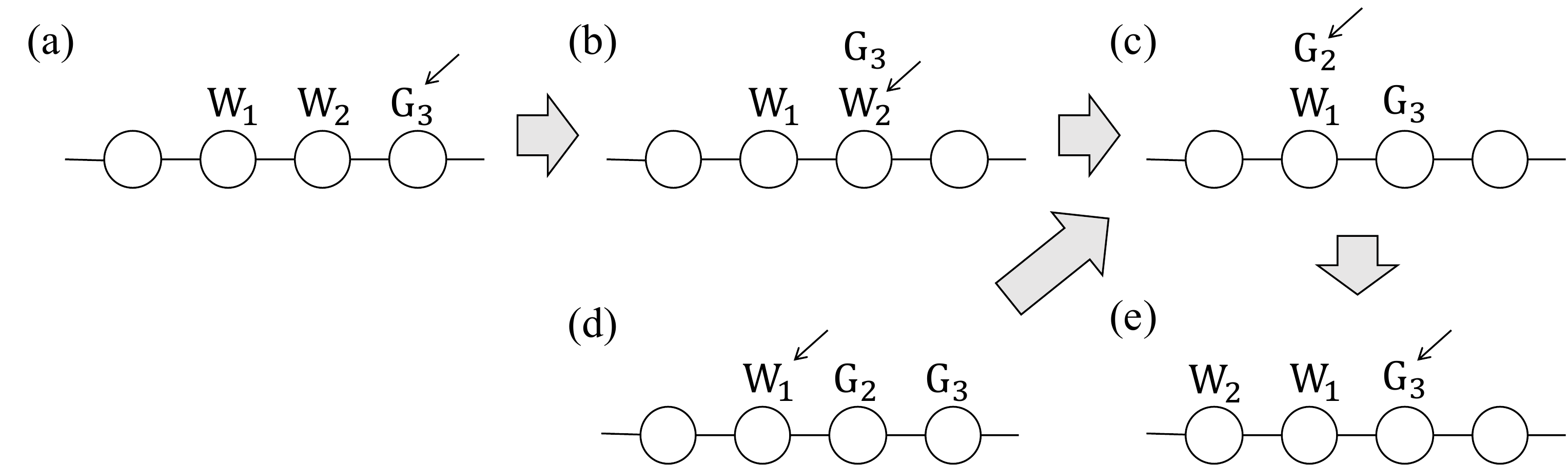}
\caption{Executions of Algorithm \ref{alg:ap3}}
\label{fig:alg-ap}
\end{center}
\end{figure}

\begin{theorem}
\label{thm:ap-alg-cor}
In case of $\phi=1$ and $k=3$, Algorithm \ref{alg:ap3} solves perpetual exploration from initial configurations
\[
\W\W\G, \W\G\G, \G\W\W, \G\G\W, 
\TOWERA{}{\W}\TOWERA{\G}{\W}, \TOWERA{\G}{\W}\TOWERA{}{\W}, 
\TOWERA{}{\G}\TOWERA{\W}{\G}, \ \mbox{and}\  \TOWERA{\W}{\G}\TOWERA{}{\G}
\]
for $n\ge 3$ in the ASYNC model.
\end{theorem}

We can also show that other initial configurations are unsolvable for $n\ge 9$ in the SSYNC model. This implies Algorithm \ref{alg:ap3} is universal with respect to perpetual exploration for $n\ge 9$ in the SSYNC and ASYNC models.

\begin{theorem}
\label{thm:universal-async}
In case of $\phi=1$, $k=3$, and $Col=\{\G,\W\}$, Algorithm \ref{alg:ap3} is universal with respect to perpetual exploration for $n\ge 9$ in the SSYNC and ASYNC models.
\end{theorem}

\subsection{Impossibility of terminating exploration with three robots}

In this subsection, we prove that three robots are not sufficient to achieve terminating exploration in the SSYNC model. Clearly this impossibility result holds in the ASYNC model.

\begin{theorem}
\label{thm:no-three-ssync}
In case of $\phi=1$ and $k=3$, no algorithm solves terminating exploration in the SSYNC model. This holds even if robots can use an infinite number of colors.
\end{theorem}

\begin{proof}
For contradiction, assume that such algorithm $\ALG$ exists. Consider an execution $E=C_0,C_1,\ldots$ of $\ALG$ in a $n_1$-node ring $R_1$ ($n_1\ge 9$). Let $t$ be the minimum instant such that three robots terminate or become disconnected at $C_t$. Next, for $n_2=4(t+1)$, let us consider an execution $E'=C'_0,C'_1,\ldots$ of $\ALG$ in a $n_2$-node ring $R_2$ such that the scheduler activates robots similarly to $E$. Clearly, as long as three robots keep connected, they do not recognize the difference between $R_1$ and $R_2$. Hence, in $E'$, three robots move similarly to $E$ until $C'_t$. Note that, since robots have visited at most $3(t+1)$ until $C'_t$, they must explore the remaining $t+1$ nodes. If three robots terminate at $C'_t$, they do not achieve exploration. If three robots become disconnected at $C'_t$, robots cannot achieve exploration from the proof of Theorem \ref{thm:universal-async}.
\end{proof}

\subsection{A terminating exploration algorithm for four robots with two colors}

In this subsection, we give a terminating exploration algorithm for four robots with two colors in case of $n\ge 4$. A set of colors is $Col=\{\G,\W\}$. The algorithm is given in Algorithm \ref{alg:at4}. Note that rules $0\G\W$, $0T\W$, $0T\G$ are identical to Algorithm \ref{alg:ap3}. Hence, once three robots construct sub-configurations
\[
\CONF_{pe}=\left\{\W\W\G, \TOWERA{}{\W}\TOWERA{\G}{\W}, \TOWERA{\W}{\G}\TOWERA{}{\G}\right\},
\]
they explore the ring similarly to Algorithm \ref{alg:ap3}.

\begin{algorithm}[t]
\caption{Asynchronous Terminating Exploration for $k=4$}
\label{alg:at4}
\begin{algorithmic}
\renewcommand{\algorithmicrequire}{\textbf{Initial configurations}}
\REQUIRE
\STATE $\W\W\G\G$, $\W\W\W\G$, $\W\W\G\W$, $\G\G\W\W$, $\G\W\W\W$, $\W\G\W\W$.
\renewcommand{\algorithmicrequire}{\textbf{Rules}}
\REQUIRE
\STATE $0\G\W:\ \ \TOWERN{\EMP}\TOWERN{(\G)}\TOWERN{\W} ::\ \G,\RIGHT$\vspace{3mm}
\STATE $\G\G\W:\ \ \TOWERN{\G}\TOWERN{(\G)}\TOWERN{\W} ::\ \G,\RIGHT$\vspace{3mm}
\STATE $0T\W:\ \ \TOWERA{}{\EMP}\TOWERA{\G}{(\W)}\TOWERA{}{\W} ::\ \G,\RIGHT$\vspace{3mm}
\STATE $\G T\W:\ \ \TOWERA{}{\G}\TOWERA{\G}{(\W)}\TOWERA{}{\W} ::\ \G,\RIGHT$\vspace{3mm}
\STATE $0T\G:\ \ \TOWERA{}{\EMP}\TOWERA{\W}{(\G)}\TOWERA{}{\G} ::\ \W,\LEFT$\vspace{3mm}
\STATE $0\W\G:\ \ \TOWERN{\EMP}\TOWERN{(\W)}\TOWERN{\G} ::\ \W,\RIGHT$
\end{algorithmic}
\end{algorithm}

Executions of Algorithm \ref{alg:at4} for $n\ge 5$ are given in Fig. \ref{fig:alg-at-wwgg}. At configuration $\W\W\G\G$ (Fig.\,\ref{fig:alg-at-wwgg}(a)) only $r_3$ can move by rule $\G\G\W$, and the configuration becomes one in Fig.\,\ref{fig:alg-at-wwgg}(b) after it moves. Since $r_1$, $r_2$, and $r_3$ form a sub-configuration in $\CONF_{pe}$, they explore the ring. Since $r_4$ does not move, three robots eventually join $r_4$ from the opposite direction (Fig.\,\ref{fig:alg-at-wwgg}(c)).

Figure\,\ref{fig:alg-at-last} shows executions after a configuration in Fig.\,\ref{fig:alg-at-wwgg}(c). In this figure, we reassign indices to robots: $r_1$, $r_2$, $r_3$, and $r_4$ form sub-configuration $\G\W\W\G$ in this order. Since $r_1$ and $r_4$ are enabled, robots can make several behaviors depending on activation by the scheduler. Notation LC-$i$ (resp., M-$i$) means the scheduler activates Look and Compute phases (resp., Move phase) of $r_i$. Although the scheduler can activate Look and Compute phases separately, we combine the two phases because a Look phase does not change a configuration. At configuration $\G\W\W\G$, the scheduler can activate $r_1$ and $r_4$ to change the configuration. If the scheduler activates exactly one robot, we only show the case of $r_1$ because the case of $r_4$ is symmetric to $r_1$. An arrow below a robot means that the robot decides to move to the direction. In every execution from configuration $\G\W\W\G$, robots eventually terminate. That is, Algorithm \ref{alg:at4} solves terminating exploration from initial configuration $\W\W\G\G$.

\begin{figure}[t]
\begin{center}
\includegraphics[scale=0.35]{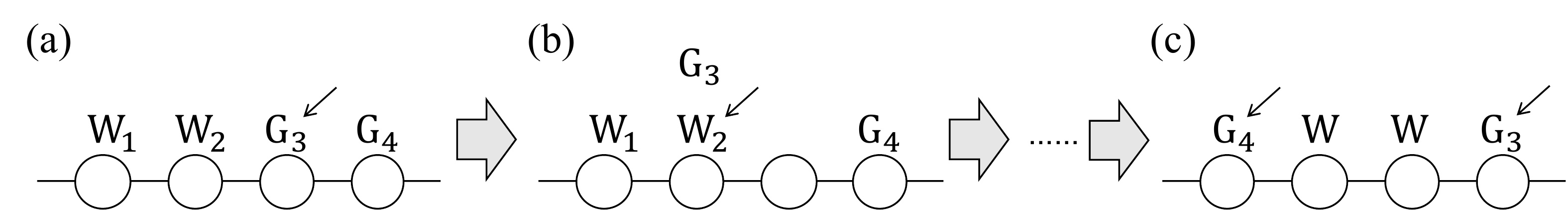}
\caption{An execution from $\W\W\G\G$ of Algorithm \ref{alg:at4}}
\label{fig:alg-at-wwgg}
\end{center}
\end{figure}

\begin{figure}[t]
\begin{center}
\includegraphics[width=0.95\textwidth]{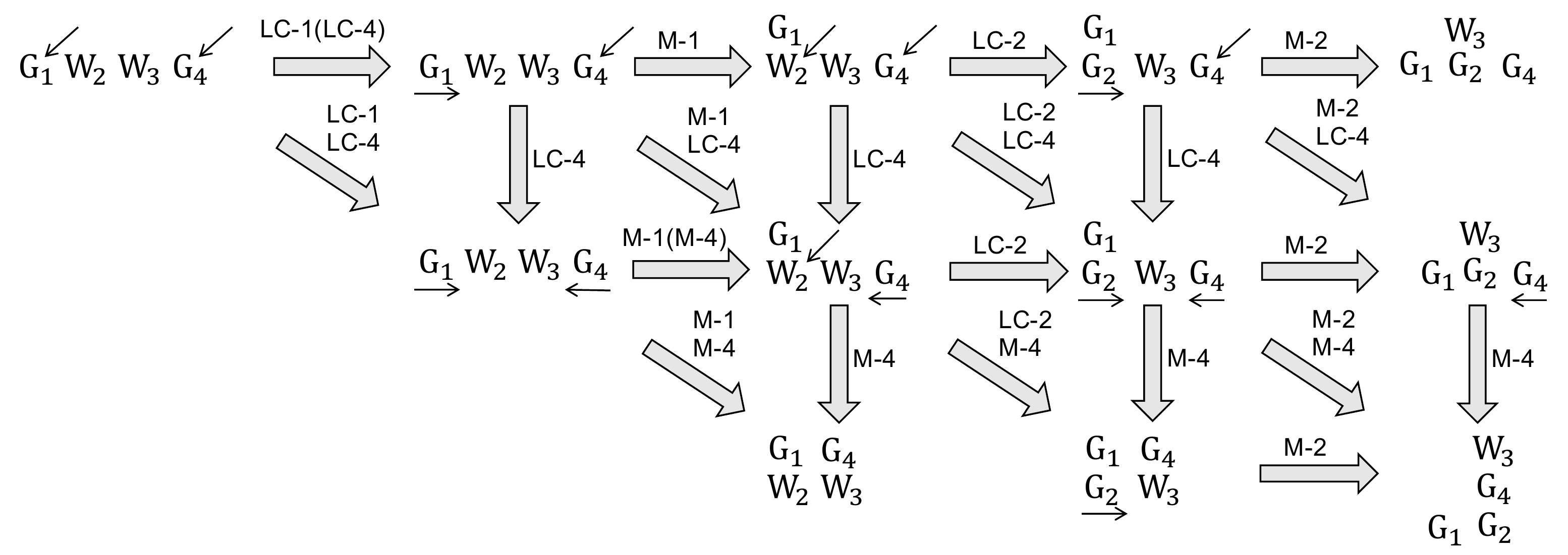}
\caption{A termination execution of Algorithm \ref{alg:at4}}
\label{fig:alg-at-last}
\end{center}
\end{figure}

We consider other initial configurations $\W\W\W\G$ and $\W\W\G\W$ in Fig.\,\ref{fig:alg-at-other}. From initial configuration $\W\W\W\G$ (Fig.\,\ref{fig:alg-at-other}(a)), robots eventually form configuration $\W\W\G\G$ (Fig.\,\ref{fig:alg-at-other}(e)) and thus they can solve terminating exploration. From initial configuration $\W\W\G\W$ (Fig.\,\ref{fig:alg-at-other}(f)), robots form a configuration in Fig.\,\ref{fig:alg-at-other}(g) and the configuration is the same as in Fig.\,\ref{fig:alg-at-other}(b). Hence, they can solve terminating exploration.

\begin{figure}[t]
\begin{center}
\includegraphics[scale=0.35]{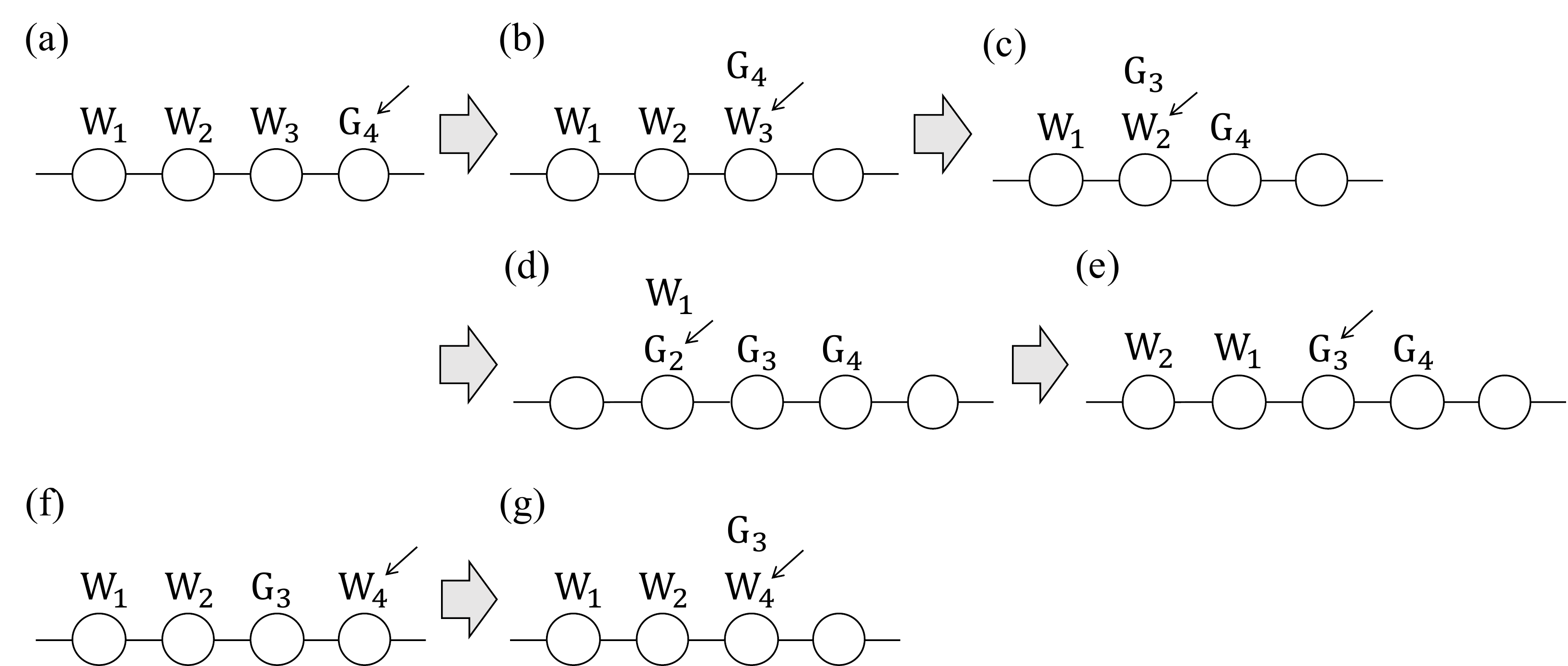}
\caption{Executions from $\W\W\W\G$ and $\W\W\G\W$ of Algorithm \ref{alg:at4}}
\label{fig:alg-at-other}
\end{center}
\end{figure}

Since configurations $\G\G\W\W$, $\G\W\W\W$, and $\W\G\W\W$ are symmetric to $\W\W\G\G$, $\W\W\W\G$, and $\W\W\G\W$, respectively, we have the following theorem.

\begin{theorem}
In case of $\phi=1$ and $k=4$, Algorithm \ref{alg:at4} solves terminating exploration from initial configurations $\W\W\G\G$, $\W\W\W\G$, $\W\W\G\W$, $\G\G\W\W$, $\G\W\W\W$, and $\W\G\W\W$ for $n\ge 5$ in the ASYNC model.
\end{theorem}

Note that we can construct another algorithm by swapping the roles of colors $\G$ and $\W$ in Algorithm \ref{alg:at4}. Clearly this algorithm solves terminating exploration from configurations such that colors $\G$ and $\W$ are swapped from solvable configurations for Algorithm \ref{alg:at4}. This implies configurations $\G\G\G\W$, $\G\G\W\G$, $\W\G\G\G$, and $\G\W\G\G$ are also solvable. Hence, we have the following lemma.

\begin{lemma}
\label{lem:at-solvable}
If $k=4$ holds and a set of colors is $\{\G,\W\}$, configurations $\W\W\G\G$, $\W\W\W\G$, $\W\W\G\W$, $\G\G\W\W$, $\G\W\W\W$, $\W\G\W\W$, $\G\G\G\W$, $\G\G\W\G$, $\W\G\G\G$, and $\G\W\G\G$ are solvable for terminating exploration in the ASYNC model.
\end{lemma}

We also prove that there exists no universal algorithm with respect to terminating exploration for four robots with two colors. This validates the assumption that Algorithm \ref{alg:at4} starts from some designated initial configuration.

\begin{theorem}
\label{thm:no-universal-ssync}
In case of $\phi=1$, $k=4$, and $\kappa=2$, no universal algorithm exists with respect to terminating exploration in the SSYNC and ASYNC models.
\end{theorem}

\begin{proof}
(Sketch) Similarly to Theorem \ref{thm:no-universal-fsync}, we can show that the distance between some pair of robots become large and some three robots must explore the ring. At that time, we also show that the three robots always form a sub-configuration in
\[
\CONF_{exp}=\left\{\W\G\G, \G\G\W, \G\W\W, \W\W\G, \TOWERA{}{\W}\TOWERA{\G}{\W}, \TOWERA{\G}{\W}\TOWERA{}{\W}, \TOWERA{}{\G}\TOWERA{\W}{\G}, \TOWERA{\W}{\G}\TOWERA{}{\G} \right\}.
\]
Note that this is a set of solvable configurations for perpetual exploration. Since a universal algorithm makes the three robots explore the ring by changing their sub-configuration from one in $\CONF_{exp}$ to another in $\CONF_{exp}$, we can show that the algorithm should include some set of rules. Lastly, we prove the set of rules makes some solvable initial configuration unsolvable.
\end{proof}

\section{Conclusions}

In this paper, we investigated the possibility of exploration algorithms for myopic luminous robots evolving in uniform ring-shaped networks. Considering weakest possible assumptions for myopia and luminosity, we proved that: \emph{(i)} in the fully synchronous model, two and three robots are necessary and sufficient to achieve perpetual and terminating exploration, respectively, and \emph{ii)} in the semi-synchronous and asynchronous models, three and four robots are necessary and sufficient to achieve perpetual and terminating exploration, respectively. These tight results characterize the power of lights for myopic robots since, without lights, five robots are necessary and sufficient to achieve terminating exploration in the fully synchronous model, and no terminating exploration algorithm exists in the semi-synchronous and asynchronous models. We also showed that our perpetual exploration algorithms are universal, and that no universal algorithm exists for terminating exploration.

This paper leaves open many issues with respect to problem solvability for myopic luminous robots. In case of non-myopic luminous robots, the difference between the semi-synchronous model and the asynchronous model disappears. Does this difference still hold for \emph{myopic} luminous robots? If visibility $\phi$ is large, robots may be able to use distance to neighboring robots to store information instead of lights. Now, is there some relation between tasks achieved by myopic luminous robots with a large number of colors, and tasks achieved by non-luminous robots with large visibility? Is there a tradeoff between the visibility distance and the number of colors?
It is also interesting to consider other tasks and topologies with myopic luminous robots.

\section*{Acknowledgments}
This work was partially supported by a mobility scholarship of the author at Sorbonne University in the frame of the Erasmus Mundus Action 2 Project TEAM Technologies for Information and Communication Technologies, funded by the European Commission. This publication reflects the view only of the authors, and the Commission cannot be held responsible for any use which may be made of the information contained therein. This work was also partially supported by JSPS KAKENHI Grant Number 18K11167. 

\bibliography{refs}


\newpage
\appendix

\section{Proof of Theorem \ref{thm:no-universal-fsync}}

To prove Theorem \ref{thm:no-universal-fsync} by contradiction, we assume universal algorithm $\ALG$ exists. We assume $Col=\{\G,\W\}$. We consider a ring with $n>4\cdot 10^3$ nodes and an execution $E=C_0,C_1,\ldots$ of $\ALG$ in the ring. Note that, since three robots can visit at most three new nodes in each configuration, $E$ includes at least $n/3>(4/3)\cdot 10^3$ configurations.

The outline of the proof is as follow. We first prove that three robots eventually become disconnected (Lemma \ref{lem:ft-imp-distance}), and at that time two robots compose a sub-configuration $\G\W$ (Lemma \ref{lem:ft-imp-conf}). Then, we prove that some rule is necessary to make the two robots explore the ring (Lemma \ref{lem:ft-imp-rules}), however the rule makes a solvable configuration unsolvable (Lemma \ref{lem:ft-imp-unsolvable}).

\begin{lemma}
\label{lem:ft-imp-distance}
The distance between some pair of two robots becomes at least five before the $10^3$-th configuration of execution $E$.
\end{lemma}

\begin{proof}
For contradiction, assume that the distance between any pair of robots is at most four until $10^3$-th configuration of $E$. For each configuration $C_i$ ($0\le i\le 10^3$), we define $C'_i$ as a sub-configuration of $C_i$ such that $C'_i$ contains five nodes (\emph{i.e.}, the length of $C'_i$ is five), every robot occupies one of the five nodes, and some robot occupies the first node of $C'_i$. Since each robot has one of two colors and occupies one of five nodes, the number of such sub-configurations is at most $10^3$. Hence, for some $t$ and $u$ ($t<u\le 10^3$), sub-configurations $C'_t$ and $C'_u$ are identical. In the following, we consider two cases.

The first case is that the five nodes contained in $C'_t$ and $C'_u$ are identical. This implies configurations $C_t$ and $C_u$ are identical. Since $\ALG$ is deterministic, robots repeat the behavior from $C_t$ to $C_u$ after $C_u$. Since robots can visit at most $3\cdot 10^3$ nodes before configuration $C_t$, they cannot visit remaining nodes after $C_u$. Thus robots cannot complete exploration.

The second case is that the five nodes contained in $C'_t$ and $C'_u$ are not identical. This means robots change their position from $C_t$ to $C_u$. However, since no robots exist out of the five nodes, the behaviors of robots depend only on the sub-configuration. Hence, robots repeat the behaviors from $C_t$ to $C_u$ after $C_u$. This implies robots continue to explore the ring, and thus they cannot terminate.

For both cases, robots cannot achieve terminating exploration. This is a contradiction.
\end{proof}

From Lemma \ref{lem:ft-imp-distance}, the distance between some pair of two robots becomes at least five before the $10^3$-th configuration. Let $t_1$ be the smallest instant such that the distance between some pair of two robots is at least five at $C_{t_1}$. Clearly $t_1 \le 10^3$ holds. Since three robots can visit at most $3 (t_1+1)$ nodes before $C_{t_1}$, at least $10^3-3$ unexplored nodes exist at $C_{t_1}$. The next two lemmas show that, to explore the unexplored nodes, robots have to construct a specific sub-configuration.

\begin{lemma}
\label{lem:ft-imp-notower}
If a tower exists at configuration $C_{t_1}$, robots break the tower before they visit three unexplored nodes.
\end{lemma}

\begin{proof}
Assume that a tower exists at configuration $C_{t_1}$. Since the number of robots is three and the distance between some pair of robots is at least five at $C_{t_1}$, the tower includes two robots and the views of the robots are symmetric. The view of a robot not in the tower is also symmetric. Hence, we can define their territories similarly to Lemma \ref{lem:imp-territory}. This implies robots do not go out of their territories before they break the tower. Therefore, the lemma holds.
\end{proof}

Let $t_2$ be the smallest instant such that the distance between some pair of two robots is at least five and no tower exists at $C_{t_2}$. Since at least $10^3-3$ unexplored nodes exist at $C_{t_1}$, at least $10^3-7$ unexplored nodes exist at $C_{t_2}$ from Lemma \ref{lem:ft-imp-notower}.

\begin{lemma}
\label{lem:ft-imp-conf}
Configuration $C_{t_2}$ includes sub-configuration $\G\W\EMP^i X$ or $\W\G\EMP^i X$ for some $i\ge 3$ and $X\in Col$ (or sub-configuration symmetric to one of them).
\end{lemma}

\begin{proof}
For contradiction, assume that $C_{t_2}$ does not include such sub-configurations. Since the distance between some pair of two robots is at least five and no tower exists, $C_{t_2}$ includes either 1) sub-configuration $Y\EMP^i Z \EMP^j X$ for some $i+j\ge 3$ and $X,Y,Z\in Col$, 2) sub-configuration $YY\EMP^i X$ for some $i\ge 3$ and $X,Y\in Col$, or 3) sub-configuration symmetric to Case 1 or 2.

In Case 1, we can define an independent territory set. This implies robots cannot visit remaining unexplored nodes from Lemma \ref{lem:imp-territory}. Thus, they cannot achieve exploration.

Let us consider Case 2. We assume robots $r_1$ and $r_2$ have color $Y$ and $r_3$ has color $X$ at $C_{t_2}$. Since the scheduler can make $r_3$ move forward and backward, $r_1$ and $r_2$ should move to achieve exploration. Since the view of $r_1$ is identical (symmetric) to the view of $r_2$, $r_1$ and $r_2$ make symmetric behaviors. If $r_1$ moves toward $r_2$, $r_1$ and $r_2$ swap their positions and visit no unexplored nodes. If $r_1$ moves against $r_2$, its sub-configuration becomes $Y'\EMP\EMP Y'\EMP^{j} X'$ for some $j\in\{i-2,i-1,i\}$ and $X',Y'\in Col$. Here $Y'$ is a new color of $r_1$ and $r_2$, and $j$ and $X'$ depend on the behavior of $r_3$. For any case, this sub-configuration reduces to Case 1. This implies robots cannot achieve exploration.

We consider Case 3 similarly to Cases 1 and 2. For every case, robots cannot achieve exploration. This is a contradiction.
\end{proof}

From Lemma \ref{lem:ft-imp-conf}, at configuration $C_{t_2}$, two robots with colors $\G$ and $\W$ are away from the other robot. Since the remaining isolated robot cannot explore by itself, two robots with colors $\G$ and $\W$ explore the remaining part of the ring. The following lemma proves that $\ALG$ includes some rules to realize this behavior.

\begin{lemma}
\label{lem:ft-imp-rules}
We consider the following rules:
\begin{itemize}
    \item $0\W\G X: \EMP (\W) \G:: X,\RIGHT$ (for some $X\in Col$)
    \item $0\G\W X: \EMP (\G) \W:: X,\RIGHT$ (for some $X\in Col$)
\end{itemize}
Algorithm $\ALG$ includes a rule $0\W\G X$ or $0\G\W X$
\end{lemma}

\begin{proof}
For contradiction, assume that $\ALG$ includes neither $0\W\G X$ nor $0\G\W X$. From Lemma \ref{lem:ft-imp-conf}, at configuration $C_{t_2}$, two robots with colors $\G$ and $\W$ are away from the other robot and compose sub-configuration $\G\W$ or $\W\G$. Since sub-configurations $\G\W$ and $\W\G$ are symmetric, we consider only sub-configuration $\G\W$. To complete exploration, the two robots must explore the remaining part of the ring. 

Since $\ALG$ includes neither $0\W\G X$ nor $0\G\W X$, a robot in a sub-configuration $\G\W$ does not move toward another robot in $\G\W$. We consider three cases about the movement of robots in $\G\W$:
\begin{enumerate}
\item One or both of robots in $\G\W$ move against the other robots in $\G\W$. In this case, these two robots become isolated. At this configuration, we can define an independent territory set, and thus robots cannot achieve exploration from Lemma \ref{lem:imp-territory}.

\item Both robots in $\G\W$ change their colors and do not move. In this case, robots just swap their colors and thus the views of them are not changed. This implies robots repeatedly change their colors and never move. Therefore, robots cannot achieve exploration.

\item One robot in $\G\W$ changes its color. That is, two robots in $\G\W$ at $C_{t_2}$ compose sub-configuration $\W\W$ or $\G\G$ at $C_{t_2+1}$. Recall that, from Lemma \ref{lem:ft-imp-conf}, sub-configurations of $C_{t_2}$ include $\G\W\EMP^i X$ or $\W\G\EMP^i X$ for some $i\ge 3$ and $X\in Col$. This implies sub-configurations of $C_{t_2+1}$ include $YY\EMP^j X'$ for some $j\in\{i,i+1\}$ and some $X',Y\in Col$. Here $X'$ and $j$ depend on the behavior of a robot with color $X$ at $C_{t_2}$. As proved in Case 2 of Lemma \ref{lem:ft-imp-conf}, from sub-configuration $YY\EMP^j X'$, robots cannot achieve exploration.
\end{enumerate}

For every case, robots cannot achieve exploration. This is a contradiction.
\end{proof}

From Lemma \ref{lem:ft-imp-rules}, $\ALG$ includes a rule $0\W\G X$ or $0\G\W X$. The following lemma shows, if $\ALG$ includes rule $0\W\G X$ or $0\G\W X$, $\ALG$ cannot solve terminating exploration from some solvable configuration.

\begin{lemma}
\label{lem:ft-imp-unsolvable}
1) If $\ALG$ includes rule $0\W\G X$, $\ALG$ cannot solve terminating exploration from configuration $\W\G\W$. 2) If $\ALG$ includes rule $0\G\W X$, $\ALG$ cannot solve terminating exploration from configuration $\G\W\G$.
\end{lemma}

\begin{proof}
Without loss of generality, we consider only statement 1 (Statement 2 is obtained by swapping the roles of two colors). Assume that, initially, robots $r_1$, $r_2$, and $r_3$ construct a configuration $\W\G\W$ in this order. Since $\ALG$ includes rule $0\W\G X$, the next configuration includes a sub-configuration
\[
\TOWERA{}{Y} \TOWERA{X}{X} \ \mbox{or}\ \TOWERB{X}{X}{Y}
\]
for some $X,Y\in Col$. That is, $r_1$ and $r_3$ assign $X$ to its color, move based on rule $0\W\G X$, and makes a tower, and $r_2$ assigns $Y$ to its color and moves (or stays) based on some rule. Note that $r_1$ and $r_3$ stay at the same node with the same color. This implies, since the views of $r_1$ and $r_3$ are identical, they move together after this configuration. Similarly to Lemma \ref{lem:ft-imp-distance}, the distance between two agents becomes at least five. In this configuration, since $r_1$ and $r_3$ move together, one node is occupied by $r_1$ and $r_3$ with the same color, and another node is occupied by $r_2$. Since we can define an independent territory set, robots cannot achieve exploration from Lemma \ref{lem:imp-territory}. That is, $\ALG$ cannot solve terminating exploration from configuration $\W\G\W$.
\end{proof}

From Lemmas \ref{lem:ft-imp-rules} and \ref{lem:ft-imp-unsolvable}, $\ALG$ cannot solve terminating exploration from configuration $\W\G\W$ or $\G\W\G$. However, from Lemma \ref{lem:ft3-solvable}, configurations $\W\G\W$ and $\G\W\G$ are solvable. This contradicts to the assumption that $\ALG$ is universal. Therefore, Theorem \ref{thm:no-universal-fsync} holds.

\section{Proof of Theorem \ref{thm:universal-async}}

To prove Theorem \ref{thm:universal-async}, we show that initial configurations other than ones in Theorem \ref{thm:ap-alg-cor} are unsolvable for $n\ge 9$ in the SSYNC model. First we prove the following simple lemma.

\begin{lemma}
\label{lem:imp-onenode}
Consider a configuration $C$ such that, for some node $v_i$, $v_i$ is occupied by at least one robot, all robots on $v_i$ have the same color, and $v_{i-1}$ and $v_{i+1}$ are occupied by no robot. After configuration $C$, the following statements hold.
\begin{itemize}
    \item[(1)] If no robots appear in $v_{i-2}$ and $v_{i+1}$, robots on $v_i$ cannot visit nodes other than $v_{i-1}$ and $v_i$. 
    \item[(2)] If no robots appear in $v_{i+2}$ and $v_{i-1}$, robots on $v_i$ cannot visit nodes other than $v_{i+1}$ and $v_i$. 
\end{itemize}
\end{lemma}

\begin{proof}
We prove only statement (1) because we can similarly prove statement (2). Since robots on $v_i$ have the same color, they have the same view. Consequently, if the scheduler repeatedly activates them at the same time, they continue to make the same behaviors. In addition, since the views of robots on $v_i$ are symmetric, the scheduler decides the directions of their movements. Consequently, robots on $v_i$ move to $v_{i-1}$ if they decide to move. From statement (1), no robots exist on $v_{i-2}$ and $v_i$, and hence the views of the robots are the same as in the previous configuration. By repeating such behaviors, the robots can visit only $v_{i-1}$ and $v_i$. Therefore, the lemma holds.
\end{proof}

In the following, we show that initial configurations other than ones in Theorem \ref{thm:ap-alg-cor} are unsolvable.

\begin{lemma}
\label{lem:sp-imp-tower}
Assume $n\ge 6$. Let $C$ be a configuration such that two or three robots with the same color stay at a single node. In this case, $C$ is unsolvable in the SSYNC model.
\end{lemma}

\begin{proof}
We consider three cases 1) three robots with the same color stay at a single node, 2) two robots with the same color stay at a single node and the distance to another robot is at least two, 3) two robots with the same color stay at a single node and the distance to another robot is one, and 4) three robots stay at a single node but one robot has a different color from other two robots. In Cases 1 and 2, robots cannot achieve exploration from Lemma \ref{lem:imp-territory}. 

To consider Cases 3 and 4, assume that $r_1$ and $r_2$ have the same color and stay at node $v$, and $r_3$ stays at $v$ or a neighbor of $v$. We consider the scheduler that repeats 1) activation of $r_1$ and $r_2$ and 2) activation of $r_3$. Since $r_1$ and $r_2$ have the same view, they make the same behavior. If $r_1$ and $r_2$ join the node with $r_3$ and they have the same color as $r_3$, this reduces to Case 1. If such a situation never happens, to achieve exploration, eventually $r_1$ and $r_2$ move against $r_3$, or $r_3$ moves against $r_1$ and $r_2$. Both cases reduce to Case 2. Therefore, the lemma holds.
\end{proof}

\begin{lemma}
\label{lem:sp-imp-dist4}
Assume $n\ge 9$. Let $C$ be a configuration such that the longest distance between two robots is at least four. In this case, $C$ is unsolvable in the SSYNC model.
\end{lemma}

\begin{proof}
We consider three robots $r_1$, $r_2$, and $r_3$. If the distance between any pair of robots is at least three at $C$, robots cannot achieve exploration from Lemma \ref{lem:imp-territory}. To consider the remaining cases, without loss of generality, we assume that the distance between $r_1$ and $r_2$ is at most two at $C$. We assume that $r_1$ occupies $v_0$ and $r_2$ occupies $v_0$, $v_1$, or $v_2$ at $C$. Let $v_d$ be the node occupied by $r_3$ at $C$. Since the distance between some pair of two robots is at least four, we assume $d\ge 4$ and $d\le n-4$ hold without loss of generality.

In the following, we consider the scheduler that activates one robot at each instant. We prove that, under this scheduler, robots cannot achieve exploration from every possible configuration. More concretely, we prove that $r_1$ and $r_2$ always stay at $v_{-1}$ (i.e., $v_{n-1}$), $v_0$, $v_1$, or $v_2$. This implies, since no robots appear in $v_{d-1}$ and $v_{d+2}$, $r_3$ cannot visit nodes other than $v_{d}$ and $v_{d+1}$ from Lemma \ref{lem:imp-onenode}. Therefore, robots cannot achieve exploration.

First consider the case that $r_2$ occupies $v_2$ at $C$. In this case, the view of every robot is symmetric. Consequently, if a robot moves, its direction is decided by the scheduler. This implies $r_1$ cannot move because, if $r_1$ moves, it moves to $v_{-1}$ (i.e., $v_{n-1}$), and from Lemma \ref{lem:imp-territory}, robots cannot achieve exploration. If $r_2$ moves, it moves to $v_1$ and this reduces the next case.

Next, consider the case that $r_2$ occupies $v_1$ at $C$. We consider two sub-cases.
\begin{itemize}
\item Case that $r_1$ can move. If $r_1$ moves to $v_{-1}$, robots cannot achieve exploration from Lemma \ref{lem:imp-territory}. Assume that $r_1$ moves to $v_1$, and let $C'$ be the resultant configuration. Then the views of $r_1$ and $r_2$ become symmetric at $C'$. This implies, if they move, the directions of their movements are decided by the schedule. Hence, if $r_1$ or $r_2$ moves, it moves to $v_0$. This makes the configuration go back to $C$ (possibly, $r_1$ and $r_2$ are swapped). If neither $r_1$ nor $r_2$ moves at $C'$, clearly robots cannot achieve exploration.
\item Case that $r_2$ can move. If $r_2$ moves to $v_2$, this reduces to the case that $r_2$ occupied $v_2$. Assume that $r_2$ moves to $v_0$, and let $C'$ be the resultant configuration. Then the views of $r_1$ and $r_2$ become symmetric at $C'$. This implies, if they move, the directions of their movements are decided by the schedule. Hence, if $r_1$ or $r_2$ moves, it moves to $v_1$. This makes the configuration go back to $C$ (possibly, $r_1$ and $r_2$ are swapped). If neither $r_1$ nor $r_2$ moves at $C'$, clearly robots cannot achieve exploration.
\end{itemize}

Next, consider the case that $r_2$ occupies $v_0$ at $C$. In this case, the views of $r_1$ and $r_2$ are symmetric, and hence, if they move, the directions of their movements are decided by the scheduler. If $r_1$ or $r_2$ moves, it moves to $v_1$. This reduces to the case that $r_1$ and $r_2$ occupy nodes $v_0$ and $v_1$ at $C$. If neither $r_1$ nor $r_2$ moves, clearly robots cannot achieve exploration.

From the above discussion, $r_1$ and $r_2$ always stay at $v_{-1}$, $v_0$, $v_1$, or $v_2$. Therefore robots cannot achieve exploration, and consequently the lemma holds.
\end{proof}

\begin{lemma}
\label{lem:sp-imp-dist3}
Assume $n\ge 9$. Let $C$ be a configuration such that the longest distance between two robots is three. In this case, $C$ is unsolvable in the SSYNC model.
\end{lemma}

\begin{proof}
Consider three robots $r_1$, $r_2$, and $r_3$. Without loss of generality, we assume that, at $C$, robots $r_1$ and $r_3$ occupy $v_0$ and $v_3$, respectively, and $r_2$ occupies $v_0$ or $v_1$.

First, we claim that $r_3$ never moves as long as nodes $v_2$ and $v_4$ are free. Otherwise, by activating $r_3$ several times after $C$, the scheduler can make $r_3$ move. Since the view of $r_3$ is symmetric, the direction of the movement is also decided by the scheduler. Hence $r_3$ moves to $v_4$ and the resultant configuration is unsolvable from Lemma \ref{lem:sp-imp-dist4}. Therefore, the above claim holds.

Consider the case that $r_2$ occupies $v_0$ at $C$. Since $r_3$ cannot move, $r_1$ or $r_2$ should move. Since the views of $r_1$ and $r_2$ are symmetric, the directions of their movements are decided by the scheduler. Hence, $r_1$ or $r_2$ moves to $v_{-1}$, and the resultant configuration is unsolvable from Lemma \ref{lem:sp-imp-dist4}. 

Consider the case that $r_2$ occupies $v_1$ at $C$. Since $r_3$ cannot move, $r_1$ or $r_2$ should move. We consider the following three sub-cases.
\begin{itemize}
\item If $r_1$ moves to $v_{-1}$, the resultant configuration is unsolvable from Lemma \ref{lem:sp-imp-dist4}. \item If $r_1$ moves to $v_1$ (resp., if $r_2$ moves to $v_0$), $r_1$ and $r_2$ occupy $v_1$ (resp., $v_0$) at the resultant configuration $C'$. Since $r_3$ cannot move at $C'$, $r_1$ or $r_2$ should move after $C'$. Since views of $r_1$ and $r_2$ are symmetric, the directions of their movements are decided by the scheduler. Hence, $r_1$ or $r_2$ moves to $v_0$ (resp., $v_1$), and this makes the configuration go back to $C$ (possibly, $r_1$ and $r_2$ are swapped). 
\item If $r_2$ moves to $v_2$, $r_1$, $r_2$, and $r_3$ occupies $v_0$, $v_2$, $v_3$, respectively. Since we do not consider colors here, the resultant configuration is symmetric to $C$ considered in the current case (i.e., $r_1$, $r_2$, and $r_3$ occupies $v_0$, $v_1$, $v_3$). 
\end{itemize}
For all sub-cases, the scheduler can make the configuration unsolvable or move robots infinitely among nodes $v_0$, $v_1$, $v_2$, and $v_3$. Hence, robots cannot achieve exploration in this case.
\end{proof}

\begin{lemma}
\label{lem:sp-imp-dist2-tower}
Assume $n\ge 9$. Let $C$ be a configuration such that the longest distance between two robots is two and two robots occupy a single node. In this case, $C$ is unsolvable in the SSYNC model.
\end{lemma}

\begin{proof}
We consider three robots $r_1$, $r_2$, and $r_3$. Without loss of generality, we assume that $r_1$ and $r_2$ occupy $v_1$ and $r_3$ occupies $v_3$. Since views of all robots are symmetric, the directions of their movements are decided by the scheduler. Hence, if $r_1$ or $r_2$ moves, it moves to $v_0$, and if $r_3$ moves, it moves to $v_4$. For both cases, the longest distance between two robots becomes three. Therefore, robots cannot achieve exploration from Lemma \ref{lem:sp-imp-dist3}.
\end{proof}

\begin{lemma}
\label{lem:sp-imp-dist2-sync}
Assume $n\ge 9$. Let $C$ be a configuration $XYX$ for some $X\in Col$ and $Y\in Col$. In this case, $C$ is unsolvable in the SSYNC model.
\end{lemma}

\begin{proof}
Without loss of generality, we assume that robots $r_1$, $r_2$, and $r_3$ have colors $X$, $Y$, and $X$ and occupy $v_1$, $v_2$, and $v_3$, respectively. We consider the scheduler that repeats the following activation: Activate $r_1$ and $r_3$ at the same time, and then activate $r_2$. 

First consider the case that $r_1$ and $r_3$ move. Since $r_1$ and $r_3$ have the same views, they make symmetric behaviors. If $r_1$ and $r_3$ move to $v_2$, $r_1$ and $r_3$ occupy $v_2$ and they have the same color. Hence, the resultant configuration is unsolvable from Lemma \ref{lem:sp-imp-tower}. If $r_1$ and $r_3$ move to $v_0$ and $v_4$, respectively, the distance between $r_1$ and $r_3$ becomes four. Hence, the configuration is unsolvable from Lemma \ref{lem:sp-imp-dist4}.

Next consider the case that $r_2$ moves. In this case, $r_2$ joins $r_1$ or $r_3$ and makes a tower. The resultant configuration is unsolvable from Lemma \ref{lem:sp-imp-dist2-tower}.
\end{proof}

From Lemmas \ref{lem:sp-imp-dist4} and \ref{lem:sp-imp-dist3}, the longest distance between two robots is at most two. Three robots must be connected from Lemma \ref{lem:sp-imp-dist2-tower}, and they must not form a symmetric configuration from Lemma \ref{lem:sp-imp-dist2-sync}. In addition, robots with the same color must not occupy a single node from Lemma \ref{lem:sp-imp-tower}. This implies all configurations other than initial configurations described in Theorem \ref{thm:ap-alg-cor} are unsolvable. Clearly, if configuration $C$ is unsolvable in the FSYNC model, $C$  is also unsolvable in the ASYNC model. Therefore, we have Theorem \ref{thm:universal-async}.

\section{Proof of Theorem \ref{thm:no-universal-ssync}}

To prove Theorem \ref{thm:no-universal-ssync}, we show the following lemma.

\begin{lemma}
\label{lem:no-universal-lemma}
Assume that $\phi=1$, $k=4$, and $Col=\{\G,\W\}$. We define a set of configurations $\CONF_{sol}$ as follows:
\begin{eqnarray*}
\CONF_{sol}&=&\{ \W\W\G\G, \W\W\W\G, \W\W\G\W, \G\G\W\W, \G\W\W\W, \\
           & &\hspace{2mm}\W\G\W\W, \G\G\G\W, \G\G\W\G, \W\G\G\G, \G\W\G\G\}.
\end{eqnarray*}
In the SSYNC model, no algorithm exists that solves terminating exploration from any initial configuration in $\CONF_{sol}$.
\end{lemma}

From Lemma \ref{lem:at-solvable}, $\CONF_{sol}$ is a subset of solvable initial configurations in the SSYNC and ASYNC models. Hence, universal algorithms for the SSYNC model must solve terminating exploration from any initial configuration in $\CONF_{sol}$ in the SSYNC model. In addition, universal algorithms for the ASYNC model must also solve terminating exploration from any initial configuration in $\CONF_{sol}$ in the SSYNC model. However, Lemma \ref{lem:no-universal-lemma} implies no such algorithms exist. This implies no universal algorithm exists in the SSYNC and ASYNC model. Therefore, we have Theorem \ref{thm:no-universal-ssync}.

In the rest of this section, we prove Lemma \ref{lem:no-universal-lemma} by contradiction. Assume that there exists an algorithm $\ALG$ that solves terminating exploration from any initial configuration in $\CONF_{sol}$ in the SSYNC model. Let $d_1=10^3+1$, $d_2=3d_1$, and $d_3=(2d_2)^4$. We consider a ring with $n>5d_3$ nodes and an execution $E=C_0,C_1,\ldots$ of $\ALG$ in the ring. Here we assume that, in each instant, the scheduler activates at least one robot that changes its position or its color unless $\ALG$ terminates. That is, $C_i\neq C_{i+1}$ holds for every $i\le t$, where $t$ is the instant such that $\ALG$ terminates in $C_t$. Note that, since four robots can visit at most four new nodes in each configuration, $E$ includes at least $n/4>(5/4) d_3$ configurations.

The outline of the proof is as follows. In Lemmas \ref{lem:ssync-unsolvable-1} and \ref{lem:ssync-unsolvable-2}, we give unsolvable configurations. Then, we prove that the distance between some pair of two robots becomes at least $d_2$ (Lemma \ref{lem:at-imp-distance}). We also prove that, in Lemmas \ref{lem:at-imp-perp} and \ref{lem:at-imp-form}, that some three robots must explore the ring by always forming a sub-configuration in
\[
\CONF_{exp}=\left\{\W\G\G, \G\G\W, \G\W\W, \W\W\G, \TOWERA{}{\W}\TOWERA{\G}{\W}, \TOWERA{\G}{\W}\TOWERA{}{\W}, \TOWERA{}{\G}\TOWERA{\W}{\G}, \TOWERA{\W}{\G}\TOWERA{}{\G} \right\}.
\]
Note that these sub-configurations are identical to ones used in the perpetual exploration algorithm in Section \ref{sec:async-perp}. The above fact limits a possible set of rules used in algorithm $\ALG$. Lastly, we prove that, for every possible set of rules, the set of rules makes some configuration in $\CONF_{sol}$ unsolvable (Lemmas \ref{lem:at-imp-last1} and \ref{lem:at-imp-last2}). This contradicts to the assumption that $\ALG$ solves terminating exploration from any initial configuration in $\CONF_{sol}$.

First, we give examples of unsolvable configurations.
\begin{lemma}
\label{lem:ssync-unsolvable-1}
Let $C$ be a configuration that contains a sub-configuration in
\[
\CONF_{sym}=\left\{\W\G\G\W,\G\W\W\G,\W\W\W\W,\G\G\G\G,\TOWERA{\W}{\W}\TOWERA{\W}{\W},
\TOWERA{\G}{\G}\TOWERA{\G}{\G}, \TOWERA{\W}{\G}\TOWERA{\W}{\G} \right\}.
\]
In this case, $C$ is unsolvable.
\end{lemma}

\begin{proof}
Note that configurations in $\CONF_{sym}$ are symmetric and the axis of symmetry goes through a link between robots. We consider the scheduler that always activates two symmetric robots. Clearly, all robots keep their symmetry. Hence, to explore the ring, one or two robots must continue to move in one direction (and other one or two robots move in another direction). However, as shown in Lemma \ref{lem:imp-onenode} and Theorem \ref{thm:sp-imp-two}, one or two robots cannot continue to move in one direction. Therefore, robots cannot achieve exploration from configurations in $\CONF_{sym}$.
\end{proof}

\begin{lemma}
\label{lem:ssync-unsolvable-2}
Let $C$ be a configuration such that two robots with the same color occupy a single node. In this case, $C$ is unsolvable.
\end{lemma}

\begin{proof}
Let $r_1$ and $r_2$ be robots that have the same color and occupy a single node. We consider the scheduler that always activates $r_1$ and $r_2$ at the same time. Clearly, since $r_1$ and $r_2$ have the same view, they make the same behaviors. This means $r_1$ and $r_2$ move as if they were a single robot. Therefore, similarly to Theorem \ref{thm:no-three-ssync}, robots cannot achieve exploration from configuration $C$.
\end{proof}

Next, we consider an execution $E$ of algorithm $\ALG$, and prove that the distance between some pair of two robots becomes large in $E$.

\begin{lemma}
\label{lem:at-imp-distance}
The distance between some pair of two robots becomes at least $d_2$ before the $d_3=(2d_2)^4$-th configuration of execution $E$.
\end{lemma}

\begin{proof}
For contradiction, assume that the distance between any pair of robots is at most $d_2-1$ until $d_3$-th configuration of $E$. For each configuration $C_i$ ($0\le i\le d_3$), we define $C'_i$ as a sub-configuration of $C_i$ such that $C'_i$ contains $d_2$ nodes (\emph{i.e.}, the length of $C'_i$ is $d_2$), every robot occupies one of the $d_2$ nodes, and some robot occupies the first node of $C'_i$. Since each robot has one of two colors and occupies one of $d_2$ nodes, the number of such sub-configurations is at most $d_3$. Hence, for some $t$ and $u$ ($t<u\le d_3$), sub-configurations $C'_t$ and $C'_u$ are identical. In the following, we consider two cases.

The first case is that the $d_2$ nodes contained in $C'_t$ and $C'_u$ are identical. This implies configurations $C_t$ and $C_u$ are identical. Since $\ALG$ is deterministic, robots repeat the behavior from $C_t$ to $C_u$ after $C_u$. Since robots can visit at most $4d_3$ nodes before configuration $C_t$, they cannot visit remaining nodes after $C_u$. Thus robots cannot complete exploration.

The second case is that the $d_2$ nodes contained in $C'_t$ and $C'_u$ are not identical. This means robots change their positions from $C_t$ to $C_u$. However, since no robots exist out of the $d_2$ nodes, the behaviors of robots depend only on the sub-configuration. Hence, robots repeat the behaviors from $C_t$ to $C_u$ after $C_u$. This implies that robots continue to explore the ring, and thus they cannot terminate.

For both cases, robots cannot achieve terminating exploration. This is a contradiction.
\end{proof}

From Lemma \ref{lem:at-imp-distance}, the distance between some pair of two robots becomes at least $d_2$ before the $d_3$-th configuration. Let $t_1$ be the smallest instant such that the distance between some pair of two robots is at least $d_2$ at $C_{t_1}$. Clearly $t_1\le d_3$ holds. Since four robots can visit at most $4(t_1+1)$ nodes before configuration $C_{t_1}$, at least $d_3-4$ unexplored nodes exist at $C_{t_1}$. Next, we consider instant $t_2=t_1+d_1$. Since robots can visit at most $4d_1$ nodes during configurations $C_{t_1}$ to $C_{t_2}$, at least $d_3-4-4d_1$ unexplored nodes exist at $C_{t_2}$.

In Lemmas \ref{lem:at-imp-team} to \ref{lem:at-imp-form}, we show that there exist three robots that explore the network by forming a sub-configuration in $\CONF_{exp}$ from $C_{t_1}$ to $C_{t_2}$.

\begin{lemma}
\label{lem:at-imp-team}
Let $C$ be a configuration in $\{C_{t_1},C_{t_1+1},\ldots,C_{t_2}\}$. At configuration $C$, there exist three robots such that the distance among them is at most four.
\end{lemma}

\begin{proof}
Recall that the distance between some pair of two robots is at least $d_2$ at configuration $C_{t_1}$. Since the distance between two robots decreases by at most two after one cycle, the distance between some pair of two robots is at least $d_2-2d_1\ge d_1$. We consider four robots $r_1$, $r_2$, $r_3$, and $r_4$. Let $d\ge d_1$ be the longest distance between two robots. Without loss of generality, we assume the distance between $r_1$ and $r_4$ is $d$, $r_1$ occupies $v_0$, and $r_4$ occupies $v_d$ at $C$. We can also assume that $r_2$ and $r_3$ occupy $v_i$ and $v_j$ ($0\le i\le j\le d$), respectively, and $i\le d-j$ holds. For contradiction, we assume $j>4$.

Consider the case that $i=2$ or $i=3$ holds. Since $j>4$ and $j\le d-2$ holds, we can define an independent territory set $\{\{v_{-1},v_0\},\{v_2,v_3\},\{v_j,v_{j+1}\},\{v_d,v_{d+1}\}\}$. From Lemma \ref{lem:imp-territory}, robots cannot explore the ring.

Consider the case that $i=0$ or $i=1$ holds, that is, $r_1$ and $r_2$ are connected. First, we show that $r_1$ and $r_2$ continue to occupy nodes in $\{v_h |-2\le h\le 3\}$ as long as no robot appears on $v_{-3}$ or $v_4$. In the case of $i=0$, if $r_1$ or $r_2$ moves, it moves $v_1$ because the direction of its movement is decided by the scheduler. This reduces to the case of $i=1$. Let us consider the case of $i=1$. If $r_1$ moves to $v_1$ or $r_2$ moves to $v_0$, they make a tower. In this case, after $r_1$ or $r_2$ moves again, the configuration goes back to the previous one (possibly, $r_1$ and $r_2$ are swapped). Hence, we assume that eventually $r_1$ moves to $v_{-1}$ or $r_2$ moves to $v_2$. If $r_1$ moves to $v_{-1}$, we can define territories $\{v_{-2},v_{-1}\}$ for $r_1$ and $\{v_1,v_2\}$ for $r_2$. Similarly, if $r_2$ moves to $v_2$, we can define territories $\{v_{-1},v_{0}\}$ for $r_1$ and $\{v_2,v_3\}$ for $r_2$. Hence, as long as no robot appears on $v_{-3}$ or $v_4$, $r_1$ and $r_2$ continue to occupy nodes in $\{v_h |-2\le h\le 3\}$. Let us consider $r_3$ and $r_4$ in this case. If $r_3$ and $r_4$ are connected at $C$, we can show the same proposition as $r_1$ and $r_2$. That is, as long as no robot appears on $v_{d-4}$ or $v_{d+3}$, $r_3$ and $r_4$ continue to occupy nodes in $\{v_{h'}|d-3\le h'\le d+2\}$. Therefore, in this case, the four robots cannot achieve exploration. If $r_3$ and $r_4$ are not connected at $C$, we can define territories for $r_3$ and $r_4$ such that $r_3$ and $r_4$ do not go out of their territories. Hence, robots cannot achieve exploration.

Now we consider the remaining case $i>3$. Clearly $r_1$ continues to occupy $v_0$ or $v_{-1}$ as long as no robot appears on $v_{-2}$ or $v_{1}$, and $r_4$ continues to occupy $v_d$ or $v_{d+1}$ as long as no robot appears on $v_{d-1}$ or $v_{d+2}$. First consider the case that $r_2$ and $r_3$ are connected. As described above, $r_1$ and $r_4$ cannot go out of their current and neighboring nodes as long as $r_2$ or $r_3$ moves toward them. This implies $r_2$ or $r_3$ must move toward $r_1$ or $r_4$ to achieve exploration. Here we assume the scheduler activates either $r_2$ or $r_3$ at the same time. This implies $r_2$ or $r_3$ moves against another robot, and the distance between $r_2$ and $r_3$ becomes two. At this configuration, we can define an independent territory set and thus robots cannot achieve exploration from Lemma \ref{lem:imp-territory}. Next consider the case that $r_2$ and $r_3$ are disconnected. In this case, we can similarly define an independent territory set and thus, from Lemma \ref{lem:imp-territory}, robots cannot achieve exploration.

For all cases, robots cannot explore the ring. Hence, $j\le 4$ holds. That is, the distance among robots $r_1$, $r_2$, and $r_3$ is at most four.
\end{proof}

From Lemma \ref{lem:at-imp-team}, for each configuration $C$ in $\{C_{t_1},C_{t_1+1},\ldots,C_{t_2}\}$, there exist three robots such that the distance among them is at most four. Note that, since the distance from one of the three robots to another robot is at least $d_1$, a set of three robots within distance four does not change from $C_{t_1}$ to $C_{t_2}$. In the following, we define $r_1$, $r_2$, and $r_3$ as three robots that stay within distance four from $C_{t_1}$ to $C_{t_2}$. Let $r_4$ be another robot. In the following lemma, we show that $r_1$, $r_2$, and $r_3$ execute a perpetual exploration algorithm from $C_{t_1}$ to $C_{t_2}$.

\begin{lemma}
\label{lem:at-imp-perp}
Consider configuration $C^*$ such that $r_1$, $r_2$, and $r_3$ stay on the same node with the same color as configuration $C_{t_1}$ and $r_4$ does not exist. Then, from configuration $C^*$, robots $r_1$, $r_2$, and $r_3$ achieve a perpetual exploration.
\end{lemma}

\begin{proof}
First we consider configurations from $C_{t_1}$ to $C_{t_2}$. From Lemma \ref{lem:at-imp-team}, three robots $r_1$, $r_2$, and $r_3$ always stay within distance four. Hence, for each configuration $C_i$ for $t_1\le i\le t_2$, we can define $C'_i$ as a sub-configuration of $C_i$ such that $C'_i$ contains five nodes (\emph{i.e.}, the length of $C'_i$ is five), each of $r_1$, $r_2$, and $r_3$ occupies one of the five nodes, and some robot occupies the first node of $C'_i$. Since each robot has one of two colors and occupies one of five nodes, the number of such sub-configurations is at most $10^3<d_1$. Hence, for some $g$ and $h$ ($t_1\le g<h\le t_2$), sub-configurations $C'_g$ and $C'_h$ are identical. We consider two cases.

The first case is that the five nodes contained in $C'_g$ and $C'_h$ are identical. This implies that $r_1$, $r_2$, and $r_3$ do not change their positions and colors in $C_g$ and $C_h$. Hence, unless $r_1$, $r_2$, or $r_3$ observes $r_4$, they repeat the behavior from $C_g$ to $C_h$ after $C_h$ and consequently continue to visit the same nodes. On the other hand, $r_4$ is far from the three robots and cannot go out of its territory (\emph{i.e.}, its current and neighboring nodes). This implies that $r_1$, $r_2$, and $r_3$ do not observe $r_4$, and thus they cannot visit the remaining unexplored nodes.

Let us consider another case, that is, the five nodes contained in $C'_g$ and $C'_h$ are different. This implies $r_1$, $r_2$, and $r_3$ change their positions from $C_g$ to $C_h$, that is, they move forward or backward in the ring. Clearly, the three robots do not observe $r_4$ from $C_{t_1}$ to $C_h$. This implies that, if $r_4$ does not exist, they repeat the behavior from $C_g$  to $C_h$ and continue to explore the ring. That is, they can achieve a perpetual exploration. Therefore, the lemma holds.
\end{proof}

\begin{lemma}
\label{lem:at-imp-form}
From $C_{t_1}$ to $C_{t_2}$, robots $r_1$, $r_2$, and $r_3$ form a sub-configuration in $\CONF_{exp}$.
\end{lemma}

\begin{proof}
For contradiction, assume that $r_1$, $r_2$, and $r_3$ form a sub-configuration $C_s$ not in $\CONF_{exp}$. From Lemma \ref{lem:at-imp-perp}, if $r_4$ does not exist, the three robots achieve perpetual exploration from configuration $C_s$. However, $C_s$ is unsolvable for perpetual exploration from Lemmas \ref{lem:sp-imp-tower}, \ref{lem:sp-imp-dist4}, \ref{lem:sp-imp-dist3}, \ref{lem:sp-imp-dist2-tower}, and \ref{lem:sp-imp-dist2-sync}. This is a contradiction.
\end{proof}

From Lemmas \ref{lem:at-imp-perp} and \ref{lem:at-imp-form}, during configurations $C_{t_1}$ to $C_{t_2}$, robots $r_1$, $r_2$, and $r_3$ continue to form a sub-configuration in $\CONF_{exp}$ and explore the ring. That is, they must change their colors and positions from a sub-configuration to another sub-configuration in $\CONF_{exp}$. We show transitions among sub-configurations in Fig.\,\ref{fig:imp-transition}, and list all rules to realize such transitions as follows. Note that, in Fig.\,\ref{fig:imp-transition}, we do not distinguish symmetric sub-configurations such as $\G\W\W$ and $\W\W\G$.

\begin{figure}[t]
\begin{center}
\includegraphics[scale=0.35]{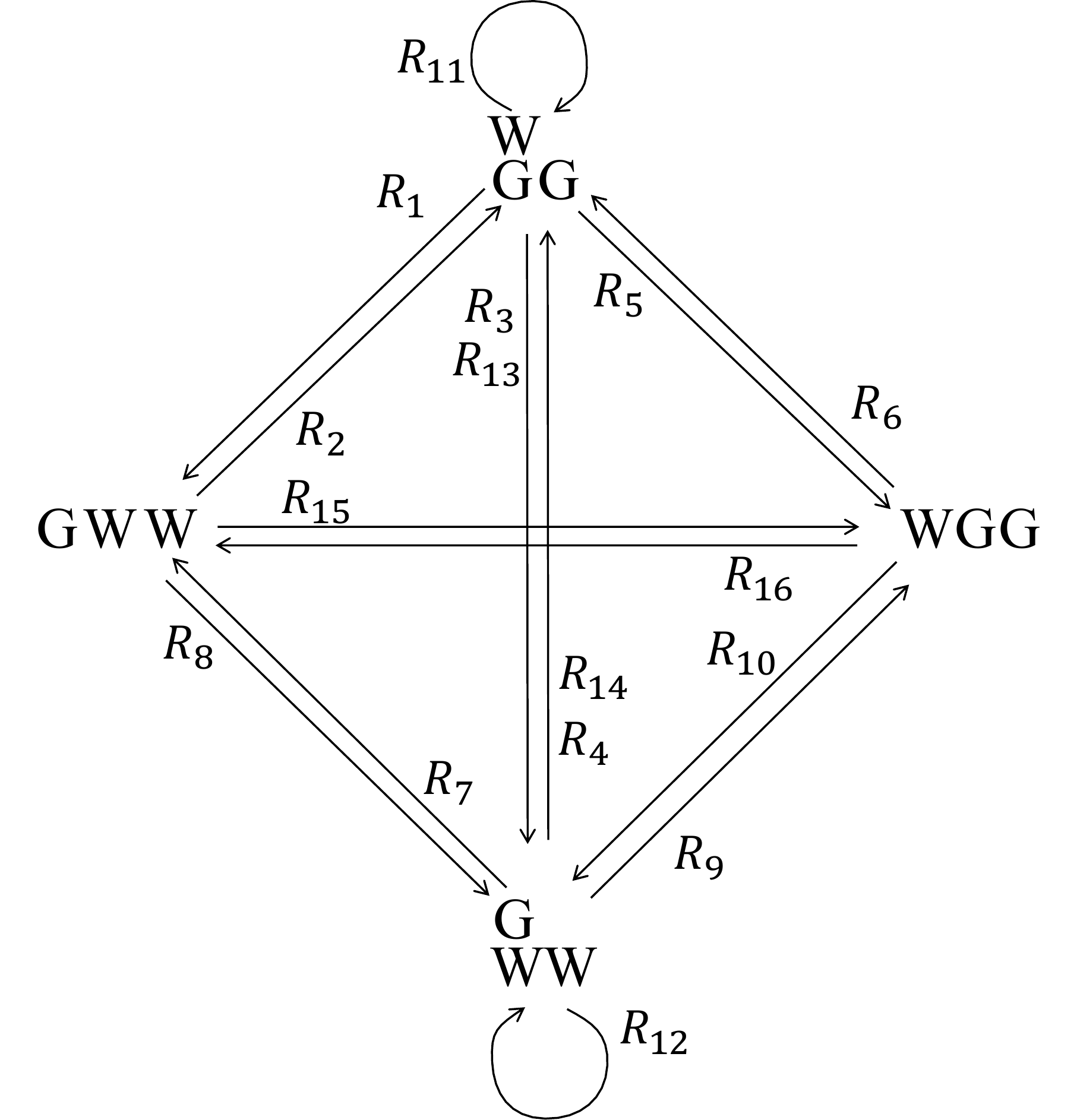}
\caption{Transitions among configurations in $\CONF_{exp}$} 
\label{fig:imp-transition}
\end{center}
\end{figure}

\begin{tabular}{ll}
     $\RULE_{1} : \TOWERA{}{\EMP}\TOWERA{\W}{(\G)}\TOWERA{}{\G}
                :: \W,\LEFT$ \hspace{1cm} &
     $\RULE_{2} : \EMP(\W)\W :: \G,\RIGHT$ \vspace{3mm}\\
     $\RULE_{3} : \TOWERA{}{\EMP}\TOWERA{\W}{(\G)}\TOWERA{}{\G}
                :: \W,\RIGHT$ & 
     $\RULE_{4} : \TOWERA{}{\EMP}\TOWERA{\G}{(\W)}\TOWERA{}{\W}
                :: \G,\RIGHT$ \vspace{3mm}\\
     $\RULE_{5} : \TOWERA{}{\EMP}\TOWERA{\G}{(\W)}\TOWERA{}{\G}
                :: \W,\LEFT$ &
     $\RULE_{6} : \EMP(\W)\G :: \W,\RIGHT$ \vspace{3mm}\\
     $\RULE_{7} : \TOWERA{}{\EMP}\TOWERA{\W}{(\G)}\TOWERA{}{\W}
                :: \G,\LEFT$ &
     $\RULE_{8} : \EMP(\G)\W :: \G,\RIGHT$ \vspace{3mm}\\
     $\RULE_{9} : \TOWERA{}{\EMP}\TOWERA{\G}{(\W)}\TOWERA{}{\W}
                :: \G,\LEFT$ &
     $\RULE_{10} : \TOWERN{\EMP}\TOWERN{(\G)}\TOWERN{\G}
                 :: \W,\RIGHT$ \vspace{3mm}\\
     $\RULE_{11} : \TOWERA{}{\EMP}\TOWERA{\G}{(\W)}\TOWERA{}{\G}
                 :: \W,\RIGHT$ &
     $\RULE_{12} : \TOWERA{}{\EMP}\TOWERA{\W}{(\G)}\TOWERA{}{\W}
                 :: \G,\RIGHT$ \vspace{3mm}\\
     $\RULE_{13} : \TOWERA{}{\EMP}\TOWERA{}{(\G)}\TOWERA{\W}{\G}
                 :: \W,\bot$ &
     $\RULE_{14} : \TOWERA{}{\EMP}\TOWERA{}{(\W)}\TOWERA{\G}{\W}
                 :: \G,\bot$ \vspace{5mm}\\
     $\RULE_{15} : \TOWERN{\G}\TOWERN{(\W)}\TOWERN{\W} :: \G,\bot$ &
     $\RULE_{16} : \TOWERN{\G}\TOWERN{(\G)}\TOWERN{\W} :: \W,\bot$ \vspace{4mm}
\end{tabular}

The next lemma shows $\ALG$ includes neither $\RULE_{15}$ nor $\RULE_{16}$.

\begin{lemma}
\label{lem:st-imp-1516}
Algorithm $\ALG$ includes neither $\RULE_{15}$ nor $\RULE_{16}$.
\end{lemma}

\begin{proof}
Assume that $\ALG$ includes $\RULE_{15}$ for contradiction. Consider a configuration $\W\G\W\W\in\CONF_{sol}$. This configuration transits to $\W\G\G\W$ by $\RULE_{15}$. From Lemma \ref{lem:ssync-unsolvable-1}, configuration $\W\G\G\W$ is unsolvable. This is a contradiction. Similarly we can show that $\ALG$ does not include $\RULE_{16}$.
\end{proof}

Since three robots that form a sub-configuration in $\CONF_{exp}$ have to explore the ring, they repeat state transitions by the above rules. This means $\ALG$ includes a set of rules that form a (simple) cycle in Fig.\,\ref{fig:imp-transition}. We list all such sets of rules as follows. Note that, from Lemma \ref{lem:st-imp-1516}, each set includes neither $\RULE_{15}$ nor $\RULE_{16}$.

\begin{tabular}{l}
    $\{\RULE_{11}\}$, $\{\RULE_{12}\}$  \\
    $\{\RULE_{1},\RULE_{2}\}$, $\{\RULE_{3},\RULE_{4}\}$,
    $\{\RULE_{3},\RULE_{14}\}$, $\{\RULE_{13},\RULE_{4}\}$,
    $\{\RULE_{13},\RULE_{14}\}$ \\ 
    $\{\RULE_{5},\RULE_{6}\}$,
    $\{\RULE_{7},\RULE_{8}\}$, $\{\RULE_{9},\RULE_{10}\}$ \\
    $\{\RULE_{2},\RULE_{3},\RULE_{7}\}$, $\{\RULE_{2},\RULE_{13},\RULE_{7}\}$,
    $\{\RULE_{8},\RULE_{4},\RULE_{1}\}$,
    $\{\RULE_{8},\RULE_{14},\RULE_{1}\}$, \\
    $\{\RULE_{10},\RULE_{4},\RULE_{5}\}$, $\{\RULE_{10},\RULE_{14},\RULE_{5}\}$,
    $\{\RULE_{6},\RULE_{3},\RULE_{9}\}$,
    $\{\RULE_{6},\RULE_{13},\RULE_{9}\}$, \\
    $\{\RULE_{2},\RULE_{5},\RULE_{10},\RULE_{7}\}$,
    $\{\RULE_{8},\RULE_{9},\RULE_{6},\RULE_{1}\}$
\end{tabular}

However, $\ALG$ cannot include some of the above sets. For example, let us consider a set $\{\RULE_{2},\RULE_{5},\RULE_{10},\RULE_{7}\}$. Starting from sub-configuration $\G\W\W$, the three robots change their sub-configuration as follows. 
\[
\TOWERA{}{\G\W\W}
\XRIGHT{\RULE_{2}}
\TOWERA{}{\G}\TOWERA{\W}{\G}\TOWERA{}{\EMP}
\XRIGHT{\RULE_{5}}
\TOWERA{}{\G\G\W}
\XRIGHT{\RULE_{10}}
\TOWERA{}{\EMP}\TOWERA{\G}{\W}\TOWERA{}{\W}
\XRIGHT{\RULE_{7}}
\TOWERA{}{\G\W\W}
\]
Note that these sub-configurations include the same nodes. This implies that, when the robots form $\G\W\W$ again, they do not change their positions. Hence, if $\ALG$ includes this set, the three robots cannot explore the ring. On the other hand, let us consider a set $\{\RULE_{2},\RULE_{3},\RULE_{7}\}$. Starting from sub-configuration $\G\W\W$, the three robots change their sub-configuration as follows.
\[
\TOWERA{}{\EMP \G\W\W}
\XRIGHT{\RULE_{2}} \TOWERA{}{\EMP}\TOWERA{}{\G}\TOWERA{\W}{\G}\TOWERA{}{\EMP}
\XRIGHT{\RULE_{3}}
\TOWERA{}{\EMP}\TOWERA{G}{\W}\TOWERA{}{\W}\TOWERA{}{\EMP}
\XRIGHT{\RULE_{7}}
\TOWERA{}{\G\W\W\EMP}
\]
In this case, when the robots form $\G\W\W$ again, they change their positions. Hence, by this set of rules, the three robots can explore the ring. By considering all sets of rules, the following four sets of rules allow robots to explore the ring:
\[
\{\RULE_{2},\RULE_{3},\RULE_{7}\},
\{\RULE_{8},\RULE_{4},\RULE_{1}\},
\{\RULE_{10},\RULE_{4},\RULE_{5}\},
\{\RULE_{6},\RULE_{3},\RULE_{9}\}.
\]
That is, $\ALG$ includes at least one of the above four sets of rules. However, in the following lemmas, each set of rules makes some configuration in $\CONF_{sol}$ unsolvable.

\begin{lemma}
\label{lem:at-imp-last1}
If $\ALG$ includes $\{\RULE_{2},\RULE_{3},\RULE_{7}\}$ or $\{\RULE_{10},\RULE_{4},\RULE_{5}\}$, $\ALG$ cannot solve terminating exploration from some configuration in $\CONF_{sol}$.
\end{lemma}

\begin{proof}
First, we consider the case that $\ALG$ includes $\{\RULE_{2},\RULE_{3},\RULE_{7}\}$. Let us consider a configuration $\W\G\W\W\in\CONF_{sol}$. From rules $\RULE_{2}$ and $\RULE_{3}$, the configuration changes as follows.
\[
\TOWERA{}{\W\G\W\W}
\XRIGHT{\RULE_{2}}
\TOWERA{}{\W}\TOWERA{}{\G}\TOWERA{\W}{\G}
\XRIGHT{\RULE_{3}}
\TOWERA{}{\W}\TOWERA{\W}{\G}\TOWERA{}{\W}
\]
Let $\CONF_a$ be the last configuration. Since robots cannot move from $\CONF_a$ by rules $\RULE_{2},\RULE_{3},\RULE_{7}$, $\ALG$ includes some other rules to move robots from $\CONF_a$. We consider all possibilities in the following.

\begin{itemize}

\item $\ALG$ includes $\RULE'_{1} : \TOWERA{}{\W}\TOWERA{\W}{(\G)}\TOWERA{}{\W}::\G,\BOTH$. In this case, robots change their states as follows.
\[
\TOWERA{}{\W}\TOWERA{\W}{\G}\TOWERA{}{\W}
\XRIGHT{\RULE'_{1}}
\TOWERA{}{\W}\TOWERA{}{\W}\TOWERA{\G}{\W}
\XRIGHT{\RULE_{2}}
\TOWERA{\G}{\W}\TOWERA{\G}{\W}
\]
From Lemma \ref{lem:ssync-unsolvable-1}, the last configuration is unsolvable.

\item $\ALG$ includes $\RULE'_{2} : \TOWERA{}{\W}\TOWERA{\W}{(\G)}\TOWERA{}{\W}::\W,\BOTH$. In this case, robots change their states as follows.
\[
\TOWERA{}{\W}\TOWERA{\W}{\G}\TOWERA{}{\W}
\XRIGHT{\RULE'_{2}}
\TOWERA{}{\W}\TOWERA{}{\W}\TOWERA{\W}{\W}
\]
From Lemma \ref{lem:ssync-unsolvable-2}, the last configuration is unsolvable.

\item $\ALG$ includes $\RULE'_{3} : \TOWERA{}{\W}\TOWERA{\W}{(\G)}\TOWERA{}{\W}::\W,\bot$. In this case, robots change their states as follows.
\[
\TOWERA{}{\W}\TOWERA{\W}{\G}\TOWERA{}{\W}
\XRIGHT{\RULE'_{3}}
\TOWERA{}{\W}\TOWERA{\W}{\W}\TOWERA{}{\W}
\]
From Lemma \ref{lem:ssync-unsolvable-2}, the last configuration is unsolvable.

\item $\ALG$ includes $\RULE'_{4} : \TOWERA{}{\W}\TOWERA{\G}{(\W)}\TOWERA{}{\W}::\G,\BOTH$. In this case, robots change their states as follows.
\[
\TOWERA{}{\W}\TOWERA{\W}{\G}\TOWERA{}{\W}
\XRIGHT{\RULE'_{4}}
\TOWERA{}{\W}\TOWERA{}{\G}\TOWERA{\W}{\G}
\XRIGHT{\RULE_{3}}
\TOWERA{}{\W}\TOWERA{\W}{\G}\TOWERA{}{\W}
\]
That is, the robots repeatedly change their states while keeping their positions. This implies that they cannot achieve exploration.

\item $\ALG$ includes $\RULE'_{5} : \TOWERA{}{\W}\TOWERA{\G}{(\W)}\TOWERA{}{\W}::\W,\BOTH$. In this case, robots change their states as follows.
\[
\TOWERA{}{\W}\TOWERA{\W}{\G}\TOWERA{}{\W}
\XRIGHT{\RULE'_{5}}
\TOWERA{}{\W}\TOWERA{}{\G}\TOWERA{\W}{\W}
\]
From Lemma \ref{lem:ssync-unsolvable-2}, the last configuration is unsolvable.

\item $\ALG$ includes $\RULE'_{6} : \TOWERA{}{\W}\TOWERA{\G}{(\W)}\TOWERA{}{\W}::\G,\bot$. In this case, robots change their states as follows.
\[
\TOWERA{}{\W}\TOWERA{\W}{\G}\TOWERA{}{\W}
\XRIGHT{\RULE'_{6}}
\TOWERA{}{\W}\TOWERA{\G}{\G}\TOWERA{}{\W}
\]
From Lemma \ref{lem:ssync-unsolvable-2}, the last configuration is unsolvable.

\item $\ALG$ includes $\RULE'_{7} : \TOWERA{}{\EMP}\TOWERA{}{(\W)}\TOWERA{\W}{\G}::\W,\RIGHT$. In this case, robots change their states as follows.
\[
\TOWERA{}{\W}\TOWERA{\W}{\G}\TOWERA{}{\W}
\XRIGHT{\RULE'_{7}}
\TOWERB{\W}{\W}{\G}\TOWERB{}{}{\W}
\]
From Lemma \ref{lem:ssync-unsolvable-2}, the last configuration is unsolvable.

\item $\ALG$ includes $\RULE'_{8} : \TOWERA{}{\EMP}\TOWERA{}{(\W)}\TOWERA{\W}{\G}::\G,\RIGHT$. In this case, robots change their states as follows.
\[
\TOWERA{}{\W}\TOWERA{\W}{\G}\TOWERA{}{\W}
\XRIGHT{\RULE'_{8}}
\TOWERB{\W}{\G}{\G}\TOWERB{}{}{\W}
\]
From Lemma \ref{lem:ssync-unsolvable-2}, the last configuration is unsolvable.

\item $\ALG$ includes $\RULE'_{9} : \TOWERA{}{\EMP}\TOWERA{}{(\W)}\TOWERA{\W}{\G}::\W,\LEFT$. Recall that, from $C_{t_1}$ to $C_{t_2}$, three robots change their states by rules $\RULE_{2}$, $\RULE_{3}$, and $\RULE_{7}$. Let us consider a configuration such that three robots form $\TOWERA{\G}{\W}\TOWERA{}{\W}$. In the next configuration, three robots can form $\TOWERA{\G}{\W}\TOWERA{}{\EMP}\TOWERA{}{\W}$ by rule $\RULE'_{9}$. This contradicts to Lemma \ref{lem:at-imp-form}.

\item $\ALG$ includes $\RULE'_{10} : \TOWERA{}{\EMP}\TOWERA{}{(\W)}\TOWERA{\W}{\G}::\G,\LEFT$. Let us consider a configuration such that three robots form $\TOWERA{\G}{\W}\TOWERA{}{\W}$ to explore the ring. In the next configuration, three robots can form $\TOWERA{\G}{\W}\TOWERA{}{\EMP}\TOWERA{}{\G}$ by rule $\RULE'_{10}$. This contradicts to Lemma \ref{lem:at-imp-form}.

\item $\ALG$ includes $\RULE'_{11} : \TOWERA{}{\EMP}\TOWERA{}{(\W)}\TOWERA{\W}{\G}::\G,\bot$. Let us consider a configuration such that three robots form $\TOWERA{\G}{\W}\TOWERA{}{\W}$ to explore the ring. After that, the three robots change their states as follows.
\[
\TOWERA{\G}{\W}\TOWERA{}{\W}
\XRIGHT{\RULE'_{11}}
\TOWERA{\G}{\W}\TOWERA{}{\G}
\XRIGHT{\RULE_{3}}
\TOWERA{}{\W}\TOWERA{\W}{\G}
\XRIGHT{\RULE'_{11}}
\TOWERA{}{\G}\TOWERA{\W}{\G}
\XRIGHT{\RULE_{3}}
\TOWERA{\G}{\W}\TOWERA{}{\W}
\]
That is, the robots repeatedly change their states while keeping their positions. This implies that they cannot explore the ring, and thus this contradicts to Lemma \ref{lem:at-imp-perp}.
\end{itemize}
From the above discussion, for any rule included in $\ALG$, some configuration in $\CONF_{sol}$ becomes unsolvable or exploration of three robots becomes impossible. Therefore, if $\ALG$ includes $\{\RULE_2,\RULE_3,\RULE_7\}$, $\ALG$ cannot solve terminating exploration from some configuration in $\CONF_{sol}$.

We can prove the case of $\{\RULE_{10},\RULE_4,\RULE_5\}$ similarly because, in rules $\RULE_{10}$, $\RULE_4$, and $\RULE_5$, the roles of colors are just swapped from rules $\RULE_2$, $\RULE_3$, and $\RULE_7$.
\end{proof}

\begin{lemma}
\label{lem:at-imp-last2}
If $\ALG$ includes $\{\RULE_8,\RULE_4,\RULE_1\}$ or $\{\RULE_6,\RULE_3,\RULE_9\}$, $\ALG$ cannot solve terminating exploration from some configuration in $\CONF_{sol}$
\end{lemma}

\begin{proof}
First, we consider the case that $\ALG$ includes $\{\RULE_8,\RULE_4,\RULE_1\}$. Let us consider configuration $\G\G\W\W\in\CONF_{sol}$. Assume that robots $r_1$, $r_2$, $r_3$, and $r_4$ form the sub-configuration $\G\G\W\W$ in this order. Since robots cannot move from this configuration by rules $\RULE_8$, $\RULE_4$, and $\RULE_1$, $\ALG$ includes some other rules to move robots from the configuration. We consider all possibilities in the following.

First consider the case that $r_1$ can move. That is, $\ALG$ includes a rule such that the guard is $\EMP(\G)\G$.
\begin{itemize}
\item $\ALG$ includes $\RULE''_{1} : \EMP(\G)\G::\G,\RIGHT$. In this case, robots change configuration $\G\G\W\W\in\CONF_{sol}$ to configuration $\TOWERA{\G}{\G}\TOWERA{}{\W}\TOWERA{}{\W}$, which is unsolvable from Lemma \ref{lem:ssync-unsolvable-2}.
\item $\ALG$ includes $\RULE''_{2} : \EMP(\G)\G::\W,\RIGHT$. In this case, from rules $\RULE''_{2}$ and $\RULE_{8}$, robots change configuration $\G\W\G\G\in\CONF_{sol}$ to configuration $\TOWERA{\G}{\W}\TOWERA{\G}{\W}$, which is unsolvable from Lemma \ref{lem:ssync-unsolvable-1}.
\item $\ALG$ includes $\RULE''_{3} : \EMP(\G)\G::\G,\LEFT$. In this case, from rules $\RULE''_{3}$ and $\RULE_{8}$, robots change their states from configuration $\G\W\G\G\in\CONF_{sol}$ as follows.
\[
\G\W\G\G
\XRIGHT{\RULE''_{3}}
\G\W\G\EMP \G 
\XRIGHT{\RULE_{8}}
\TOWERB{\G}{\G}{\W}\TOWERB{}{}{\EMP}\TOWERB{}{}{\EMP}\TOWERB{}{}{\G}
\]
The last configuration is unsolvable from Lemma \ref{lem:ssync-unsolvable-2}.
\item $\ALG$ includes $\RULE''_{4} : \EMP(G)G::\W,\LEFT$. In this case, from rules $\RULE''_{4}$ and $\RULE_{8}$, robots change their states from configuration $\G\W\G\G\in\CONF_{sol}$ as follows.
\[
\G\W\G\G 
\XRIGHT{\RULE''_{4}}
\G\W\G\EMP \W
\XRIGHT{\RULE_{8}}
\TOWERB{\G}{\G}{\W}\TOWERB{}{}{\EMP}\TOWERB{}{}{\EMP}\TOWERB{}{}{\W}
\]
The last configuration is unsolvable from Lemma \ref{lem:ssync-unsolvable-2}.
\item $\ALG$ includes $\RULE''_{5} : \EMP(\G)\G::\W,\bot$. In this case, robots change configuration $\G\G\G\W\in\CONF_{sol}$ to configuration $\W\G\G\W$, which is unsolvable from Lemma \ref{lem:ssync-unsolvable-1}.
\end{itemize}

Next consider the case that $r_2$ can move. That is, $\ALG$ includes a rule such that the guard is $\G(\G)\W$.
\begin{itemize}
\item $\ALG$ includes $\RULE''_{6} : \G(\G)\W::\G,\RIGHT$. In this case, from rules $\RULE''_{6}$ and $\RULE_8$, robots change configuration $\G\G\W\G\in\CONF_{sol}$ to configuration $\TOWERB{}{}{\G}\TOWERB{}{}{\EMP}\TOWERB{\G}{\G}{\W}$, which is unsolvable from Lemma \ref{lem:ssync-unsolvable-2}.
\item $\ALG$ includes $\RULE''_{7} : \G(\G)\W::\W,\RIGHT$. In this case, robots change configuration $\G\G\W\G\in\CONF_{sol}$ to configuration $\TOWERA{}{\G}\TOWERA{}{\EMP}\TOWERA{\W}{\W}\TOWERA{}{\G}$, which is unsolvable from Lemma \ref{lem:ssync-unsolvable-2}.
\item $\ALG$ includes $\RULE''_{8} : \G(\G)\W::\G,\LEFT$. In this case, robots change configuration $\G\G\W\G\in\CONF_{sol}$ to configuration $\TOWERA{\G}{\G}\TOWERA{}{\EMP}\TOWERA{}{\W}\TOWERA{}{\G}$, which is unsolvable from Lemma \ref{lem:ssync-unsolvable-2}.
\item $\ALG$ includes $\RULE''_{9} : \G(\G)\W::\W,\LEFT$. In this case, from rules $\RULE''_{9}$ and $\RULE_8$, robots change configuration $\G\G\W\G\in\CONF_{sol}$ to configuration $\TOWERA{\W}{\G}\TOWERA{}{\EMP}\TOWERA{\W}{\G}$. Since two robots with the same color have the same view, they make the same behaviors if they are activated at the same time. In addition, the view is symmetric, the direction of the movement is decided by the scheduler. Hence, when some robot moves from the configuration, two robots move to the middle node of the two towers and create a tower. Since they have the same color, the configuration is unsolvable from Lemma \ref{lem:ssync-unsolvable-2}.
\item $\ALG$ includes $\RULE''_{10} : \G(\G)\W::\W,\bot$. In this case, robots change configuration $\G\G\W\G\in\CONF_{sol}$ to configuration $\G\W\W\G$, which is unsolvable from Lemma \ref{lem:ssync-unsolvable-1}.
\end{itemize}

Next consider the case that $r_3$ can move. That is, $\ALG$ includes a rule such that the guard is $\G(\W)\W$. In the first four cases, we consider configurations from $C_{t_1}$ to $C_{t_2}$. During these configurations, three robots change their states by rules $\RULE_8$, $\RULE_4$, and $\RULE_1$. In particular, we consider a sub-configuration $\G\W\W$ in $\CONF_{exp}$.
\begin{itemize}
\item $\ALG$ includes $\RULE''_{11} : \G(\W)\W::\G,\RIGHT$. From a sub-configuration $\G\W\W$, the robots change their states to $\TOWERA{}{\G}\TOWERA{}{\EMP}\TOWERA{\G}{\W}$. This sub-configuration is not in $\CONF_{exp}$, which contradicts to Lemma \ref{lem:at-imp-form}.
\item $\ALG$ includes $\RULE''_{12} : \G(\W)\W::\W,\RIGHT$. From a sub-configuration $\G\W\W$, the robots change their states to $\TOWERA{}{\G}\TOWERA{}{\EMP}\TOWERA{\W}{\W}$. This sub-configuration is not in $\CONF_{exp}$, which contradicts to Lemma \ref{lem:at-imp-form}.
\item $\ALG$ includes $\RULE''_{13} : \G(\W)\W::\G,\LEFT$. From a sub-configuration $\G\W\W$, the robots change their states to $\TOWERA{\G}{\G}\TOWERA{}{\EMP}\TOWERA{}{\W}$. This sub-configuration is not in $\CONF_{exp}$, which contradicts to Lemma \ref{lem:at-imp-form}.
\item $\ALG$ includes $\RULE''_{14} : \G(\W)\W::\W,\LEFT$. From a sub-configuration $\G\W\W$, the robots change their states to $\TOWERA{\W}{\G}\TOWERA{}{\EMP}\TOWERA{}{\W}$. This sub-configuration is not in $\CONF_{exp}$, which contradicts to Lemma \ref{lem:at-imp-form}.
\item $\ALG$ includes $\RULE''_{15} : \G(\W)\W::\G,\bot$. In this case, robots change configuration $\W\G\W\W\in\CONF_{sol}$ to configuration $\W\G\G\W$, which is unsolvable from Lemma \ref{lem:ssync-unsolvable-1}.
\end{itemize}

Lastly consider the case that $r_4$ can move. That is, $\ALG$ includes a rule such that the guard is $\EMP(\W)\W$. Similarly to the previous case, we consider a sub-configuration $\G\W\W$ in $\CONF_{exp}$.
\begin{itemize}
\item $\ALG$ includes $\RULE''_{16} : \EMP(\W)\W::\G,\RIGHT$. From rules $\RULE''_{16}$ and $\RULE_1$, robots change their states from a sub-configuration $\G\W\W$ as follows.
\[
\G\W\W
\XRIGHT{\RULE''_{16}}
\TOWERA{}{\G}\TOWERA{\W}{\G}\TOWERA{}{\EMP}
\XRIGHT{\RULE_1}
\G\W\W
\]
That is, the robots repeatedly change their states while keeping their positions. This implies that they cannot explore the ring, and thus this contradicts to Lemma \ref{lem:at-imp-perp}.
\item $\ALG$ includes $\RULE''_{17} : \EMP(\W)\W::\W,\RIGHT$. From a sub-configuration $\G\W\W$, the robots change their states to $\TOWERA{}{\G}\TOWERA{\W}{\W}$. This sub-configuration is not in $\CONF_{exp}$, which contradicts to Lemma \ref{lem:at-imp-form}.
\item $\ALG$ includes $\RULE''_{18} : \EMP(\W)\W::\G,\LEFT$. From a sub-configuration $\G\W\W$, the robots change their states to $\G\W\EMP \G$. This sub-configuration is not in $\CONF_{exp}$, which contradicts to Lemma \ref{lem:at-imp-form}.
\item $\ALG$ includes $\RULE''_{19} : \EMP(\W)\W::\W,\LEFT$. From a sub-configuration $\G\W\W$, the robots change their states to $\G\W\EMP \W$. This sub-configuration is not in $\CONF_{exp}$, which contradicts to Lemma \ref{lem:at-imp-form}.
\item $\ALG$ includes $\RULE''_{20} : \EMP(\W)\W::\G,\bot$. From a sub-configuration $\G\W\W$, the robots change their states to $\G\W\G$. This sub-configuration is not in $\CONF_{exp}$, which contradicts to Lemma \ref{lem:at-imp-form}.
\end{itemize}

From the above discussion, for any rule included in $\ALG$, some configuration in $\CONF_{sol}$ becomes unsolvable or exploration of three robots becomes impossible. Therefore, if $\ALG$ includes $\{\RULE_8,\RULE_4,\RULE_1\}$, $\ALG$ cannot solve terminating exploration from some configuration in $\CONF_{sol}$

We can prove the case of $\{\RULE_6,\RULE_3,\RULE_9\}$ similarly because, in rules $\RULE_6$, $\RULE_3$, and $\RULE_9$, the roles of colors are just swapped from rules  $\RULE_8$, $\RULE_4$, and $\RULE_1$.
\end{proof}

From Lemmas \ref{lem:at-imp-last1} and \ref{lem:at-imp-last2}, $\ALG$ cannot solve terminating exploration from some configuration in $\CONF_{sol}$. Therefore, we have Lemma \ref{lem:no-universal-lemma}. As described above, this implies
Theorem \ref{thm:no-universal-ssync}.

\end{document}